\documentclass[letterpaper,twocolumn,10pt]{article}
\usepackage{usenix2019_v3}

\usepackage{graphicx}
\usepackage{subfigure}

\usepackage{tikz}
\usepackage{amsmath}
\usepackage{amssymb}
\usepackage{mathtools}
\usepackage{amsthm}
\usepackage{cleveref}
\usepackage{authblk}

\usepackage{multirow}
\usepackage{import}
\usepackage{pifont}
\usepackage{textcomp}
\usepackage{color, colortbl}
\definecolor{greyC}{RGB}{180,180,180}
\definecolor{greyL}{RGB}{235,235,235}
\usepackage{footmisc}
\usepackage[noend]{algpseudocode}
\usepackage{bm}
\usepackage[flushleft]{threeparttable}
\usepackage{makecell}

\makeatletter
\newcommand*{\addFileDependency}[1]{
  \typeout{(#1)}
  \@addtofilelist{#1}
  \IfFileExists{#1}{}{\typeout{No file #1.}}
}
\makeatother

\usepackage[utf8]{inputenc} 
\usepackage[T1]{fontenc}    
\usepackage{hyperref}       
\usepackage{url}            
\usepackage{booktabs}       
\usepackage{amsfonts}       
\usepackage{nicefrac}       
\usepackage{microtype}      
\usepackage{xcolor}         
\usepackage{adjustbox}
\usepackage{enumitem}

\usepackage{hyperref}
\usepackage{algorithm}
\usepackage{wrapfig}
\theoremstyle{plain}
\newtheorem{theorem}{Theorem}[]
\newtheorem{proposition}[theorem]{Proposition}
\newtheorem{lemma}[theorem]{Lemma}
\newtheorem{corollary}[theorem]{Corollary}
\theoremstyle{definition}
\newtheorem{definition}{Definition}[]
\newtheorem{assumption}{Assumption}[]
\newtheorem{property}{Property}[]
\newtheorem{example}{Example}[]
\theoremstyle{remark}
\newtheorem{remark}{Remark}[]

\usepackage{amsmath,amsfonts,bm}

\def\eqref#1{equation~\ref{#1}}

\def\1{\bm{1}}

\usepackage{mathtools}
\usepackage{amssymb}
\usepackage{latexsym}
\usepackage{dsfont}
\usepackage{mathrsfs}
\usepackage{amssymb}
\usepackage{amsfonts}
\usepackage{bm}
\usepackage{xspace}
\usepackage{amsthm}
\usepackage{multirow}

\newcommand*{\argmin}{\mathop{\mathrm{argmin}}}

\newcommand{\wbar}{\bar{w}}

\newcommand{\Dcal}{\mathcal{D}}

\newcommand{\Ical}{\mathcal{I}}

\newcommand{\Lcal}{\mathcal{L}}
\newcommand{\Mcal}{\mathcal{M}}

\newcommand{\Ocal}{\mathcal{O}}

\newcommand{\Ebb}{\mathbb{E}}

\ifx\BlackBox\undefined
\newcommand{\BlackBox}{\rule{1.5ex}{1.5ex}}  
\fi

\ifx\QED\undefined
\def\QED{~\rule[-1pt]{5pt}{5pt}\par\medskip}
\fi

\ifx\proof\undefined
\newenvironment{proof}{\par\noindent{\em Proof:\ }}{\hfill\BlackBox\\[.0mm]}
\fi

\ifx\theorem\undefined
\newtheorem{theorem}{Theorem}[section]
\fi

\ifx\example\undefined
\newtheorem{example}{Example}[section]
\fi

\ifx\property\undefined
\newtheorem{property}{Property}
\fi

\ifx\lemma\undefined
\newtheorem{lemma}{Lemma}[section]
\fi

\ifx\proposition\undefined
\newtheorem{proposition}{Proposition}[section]
\fi

\ifx\remark\undefined
\newtheorem{remark}{Remark}
\fi

\ifx\corollary\undefined
\newtheorem{corollary}{Corollary}[section]
\fi

\ifx\definition\undefined
\newtheorem{definition}{Definition}[section]
\fi

\ifx\conjecture\undefined
\newtheorem{conjecture}{Conjecture}[section]
\fi

\ifx\axiom\undefined
\newtheorem{axiom}[theorem]{Axiom}
\fi

\ifx\claim\undefined
\newtheorem{claim}[theorem]{Claim}
\fi

\ifx\assumption\undefined
\newtheorem{assumption}{Assumption}
\fi

\ifx\problem\undefined
\newtheorem{problem}{Problem}[section]
\fi

\ifx\fact\undefined
\newtheorem{fact}{Fact}
\fi

\newcommand{\benr}{\begin{eqnarray}}
\newcommand{\eenr}{\end{eqnarray}}
\newcommand{\benrr}{\begin{eqnarray*}}
\newcommand{\eenrr}{\end{eqnarray*}}
\newcommand{\ben}{\begin{equation}}
\newcommand{\een}{\end{equation}}
\newcommand{\benn}{\begin{equation*}}
\newcommand{\eenn}{\end{equation*}}

\begin{document}
\date{}

\title{DreamDDP: Accelerating Data Parallel Distributed LLM Training with Layer-wise Scheduled Partial Synchronization}


\author[1]{Zhenheng Tang$^{\dagger}$}
\author[2]{Zichen Tang$^{\dagger}$}
\author[2]{Junlin Huang}
\author[2]{Xinglin Pan}
\author[2]{Rudan Yan}
\author[3]{Yuxin Wang}
\author[3]{Amelie Chi Zhou} 
\author[4]{Shaohuai Shi}
\author[2,1]{Xiaowen Chu$^{\ast}$}
\author[1]{Bo Li$^{\ast}$}

\affil[1]{The Hong Kong University of Science and Technology}
\affil[2]{The Hong Kong University of Science and Technology (Guangzhou)}
\affil[3]{Hong Kong Baptist University}
\affil[4]{Harbin Institute of Technology, Shenzhen}


\maketitle
\renewcommand{\thefootnote}{}
\footnotetext[1]{$^{\dagger}$These authors contributed equally to this work.}
\footnotetext[2]{$^{\ast}$Corresponding authors: \texttt{xwchu@hkust-gz.edu.cn}, \texttt{bli@ust.hk}}
\renewcommand{\thefootnote}{\arabic{footnote}} 

\begin{abstract}
The growth of large language models (LLMs) increases challenges of accelerating distributed training across multiple GPUs in different data centers. Moreover, concerns about data privacy and data exhaustion have heightened interest in geo-distributed data centers. Communication in geo-distributed data parallel training (DDP) with stochastic gradient descent (S-SGD) is the main bottleneck in low-bandwidth environments. Local SGD mitigates communication overhead by reducing synchronization frequency, and recent studies have successfully applied it to geo-distributedly pre-train LLMs. However, we identify that its model synchronization mechanism prevents overlapping communication and computation, which makes the system lose opportunities to overlap communication and computation.

To overcome this limitation, we expand the design space of local SGD by layer-wisely decoupling model synchronization. In each iteration, only some layers are synchronized instead of the entire model after a specific number of iterations.   Leveraging this methodology, we introduce DreamDDP, a training framework to accelerate low-bandwidth distributed training with three key innovations: (1) partial local SGD with theoretical assurances of convergence rates comparable to S-SGD; (2) overlapping parameter synchronization with computation without extra GPU memory occupation; (3) identifying and exploiting three properties to schedule the communication and computation to reduce the training time based on fine-grained profiling of layer-wise communication and computation time. Empirical evaluations conducted on 32 GPUs using prominent deep learning models, including ResNet-18, ResNet-50, GPT-2, and Llama-2, demonstrate that DreamDDP enhances the convergence properties of Local SGD (and Adam) and achieves speedups ranging from $1.49\times$ to $3.91\times$ over leading baseline methods.

\end{abstract}

\section{Introduction}\label{sec:intro}

\begin{figure}[!th]
\vspace{-0.4cm}
\centering
\setlength{\abovecaptionskip}{0.cm}
\setlength{\belowcaptionskip}{-0.2cm} 
    \subfigure[Llama-2]
    {
    \includegraphics[width=0.46\linewidth]{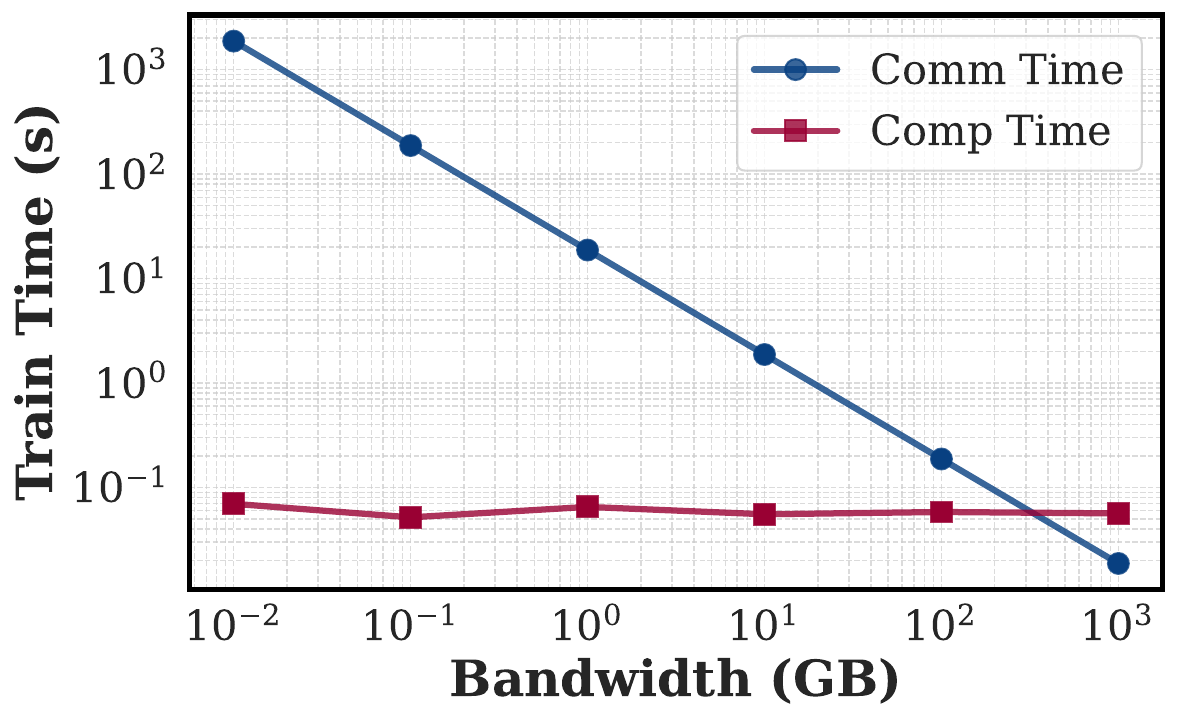}\label{fig:llama-comm-bp-time}
    }
    \subfigure[GPT-2]
    {
    \includegraphics[width=0.46\linewidth]{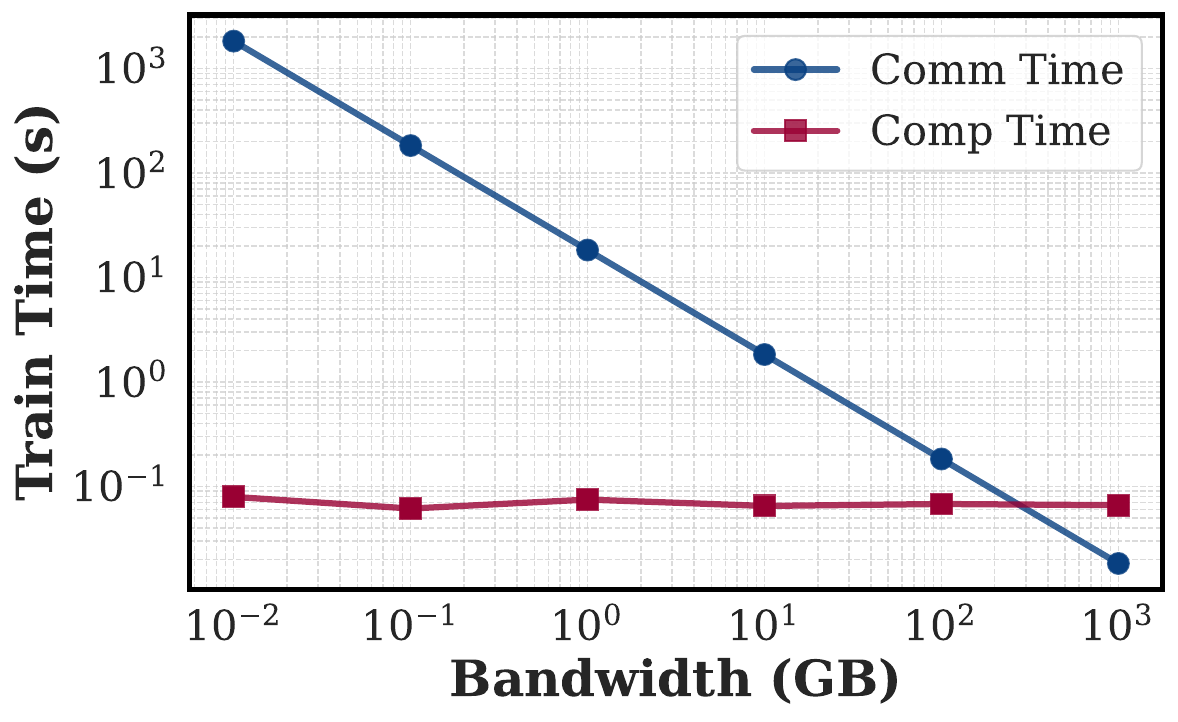}\label{fig:gpt-comm-bp-time}
    }
    \caption{Comparison of communication and computation times under varying bandwidths for GPT-2 and Llama-2.}
\label{fig:comm-bp-time}
\vspace{-0.1cm}
\end{figure}
\vspace{-0.1cm}

Modern deep learning~(DL) models often feature extremely large sizes and immense training datasets. Examples of these large language models~(LLMs) include GPT-2~\cite{gpt2}, GPT-3~\cite{gpt,chatgpt}, Gemini~\cite{reid2024gemini}, LLaMA~\cite{Touvron2023LLaMAOA}, and Mixtral~\cite{jiang2024mixtral}. Training these models incurs very high computational costs. Recent proposals suggest training LLMs across several geographically distributed data centers for scaling purposes~\cite{narayanan2021efficient,ryabinin2023swarm,tang2023fusionai,tang2024fusionllmdecentralizedllmtraining,DecentMOE,cocktailSGD,yuandecentralized,10229032,OpenDiLoCo}.



Moreover, the data exhaustion problem has been a critical issue in the training of large models~\cite{dataexhaust}. The academia and industry predict that high-quality public data will be exhausted before 2026~\cite{villalobos2022will}. Collecting high-quality private data has the privacy concerns (e.g., medical~\cite{llm_medicine} and financial~\cite{wu2023bloomberggpt} data). To address this issue, federated learning (FL) has emerged as a promising solution~\cite{openfedllm,tang2024fusefl,kairouz2019advances}. FL allows multiple parties to collaboratively train a large model without sharing their private data, which is crucial for protecting data privacy and ensuring data security.


However, as shown in Figure~\ref{fig:comm-bp-time}, the low communication bandwidth of 10 Mbps $\sim$ 1 Gbps between different data centers and multiple parties makes the communication time as the main system bottleneck~\cite{ryabinin2023swarm,DeDLOC,DecentMOE,cocktailSGD,yuandecentralized,tang2024fusionllmdecentralizedllmtraining}. For the synchronous stochastic gradient descent (S-SGD)~\cite{Bottou2016OptimizationMF} widely used in distributed training, in each training iteration the gradients are required to be synchronized between different workers, of which the communication significantly time costs. To this end, existing works that accelerate Geo-distributed training mainly focus on scheduling different data centers to conduct Data Parallel (DP) or Pipeline Parallel (PP)~\cite{ryabinin2023swarm,yuandecentralized}, or compressing gradients to reduce the communication overheads~\cite{cocktailSGD,NEURIPS2022_7a43b8eb,tang2024fusionllmdecentralizedllmtraining,tang2023fusionai}. However, the large compression ratio significantly influences the convergence rate, and the effect of communication reduction in the scheduling is limited.


To further reduce the communication overhead, local SGD (LSGD)~\cite{stich2018localsgd,woodworth2020minibatch,10229032} skips some communications by conducting local updates for $H$ times and synchronizing model parameters or pseudo gradients. Consequently, LSGD reduces communication costs by $H\times$ compared to S-SGD. In geo-distributed training, LSGD has been proven to be a practical optimization method to train the LLMs~\cite{OpenDiLoCo}. INTELLECT-1 is the first 10B parameter LLM trained in a decentralized manner with the LSGD~\cite{jaghouar2024intellect1technicalreport}. And it has been shown that LSGD achieves a similar scaling law of model parameters compared to traditional optimizers~\cite{he2024exploring}. Thus, LSGD has gradually become a new paradigm of the geo-distributed training~\cite{OpenDiLoCo,jaghouar2024intellect1technicalreport,openfedllm,douillard2024diloco,298559,qin2024federated,zhuang2023foundation,FederatedScope-LLM,10621164,sani2024photon,sunco2,cheng2024editlocalsgdbasedefficientdistributed}.

However, current LSGD requires \textbf{full synchronization}, which means to synchronize the entire model parameters after $H$ iterations, losing opportunities to overlap communication and computation tasks during backpropagation (BP) that is commonly used to accelerate the training process~\cite{203269,10.1145/3341301.3359642,9488803}. For example, WFBP~\cite{203269} communicates the gradient of each layer instantly after its BP, and overlapped with BP of subsequent layers, thus saving total training time. ASC-WFBP~\cite{9488803} further improves overlapping by simultaneous communications and scheduling. Thus, to further improve the throughputs of geo-distributed training with LSGD, the first challenge is how to design a new optimizer to avoid the limitations of the full synchronization, thus unleashing the optimization opportunities of overlapping communication and computation in LSGD like~\cite{203269,9488803}. The second challenge is how to schedule the communication and computation to reduce time costs while guaranteeing convergence. Thirdly, considering LLMs have enormous parameters and meet the memory wall in existing GPUs~\cite{ren2021zero,Rajbhandari2019ZeROMO}, the new method should not increase the GPU memory costs.

In this work, we design a novel training framework named DreamDDP to significantly accelerate DDP with improved LSGD. Firstly, we proposed \textbf{partial synchronization}, which synchronizes different parts of the model in each iteration, instead of the entire model after $H$ iterations like the loose synchronization in large-scale network systems~\cite{268961}. Thus, it offers opportunities to overlap the communication and computation of LSGD. We theoretically prove that the convergence property of partial synchronization-based local SGD (PLSGD) can be guaranteed as same as LSGD. Secondly, DreamDDP divides the neural networks into different fine-grained $L$ layers and implements \textit{in-place partial parameter synchronization} by launching layer synchronization after its BP, thus introducing no extra GPU memory occupation. Thirdly, we introduce the communication and computation profiler to profile the communication and computation time of each layer. Then, we build an overall time cost model of the FP, BP, and communication process with respect to $L$ layers and $H$ iterations of the partial synchronization and formulate it as an optimization problem. We identify three properties including \textit{optimal hiding}, \textit{delayed CO assignment}, \textit{at-least-one assignment} that can be leveraged to design a depth-first search (DFS) with significantly pruned search space to solve this optimization problem. Lastly, DreamDDP inserts more communication messages (model parameters) to fill the bubble time as long as the extra communication can be hidden by the computation, thus accelerating the training convergence.

Our contributions are summarized as follows.
\begin{itemize}
    \item We propose partial synchronization based local SGD (PLSGD) that relaxes the strong synchronization by layer-wisely decoupling model parameters, which enables the opportunity of overlapping communication and computation in LSGD to improve the training efficiency.
    \item DreamDDP profiles the real-time communication and computation time of fine-grained model layers. Then, we identify three properties leveraged by DreamDDP to design a DFS algorithm with pruned search space to optimize the throughputs of PLSGD.
    \item We theoretically prove that the DreamDDP has the same convergence rate with S-SGD, and empirically show that they have similar convergence speeds with training GPT-2, Llama-2, ResNet-18, and ResNet-50.
    \item We implement DreamDDP atop the commonly used distributed deep learning framework PyTorch Distributed~\cite{PytorchDistributed,ansel2024pytorch} and conduct real-world experiments on two different GPU clusters with 32 GPUs. Experimental results show DreamDDP achieves $1.49-3.91\times$ on iteration time improvement over the SOTA algorithms including LSGD and ASC-WFBP.
\end{itemize}

\section{Data-Parallel Distributed Training}\label{sec:prelimilary}
In this section, we present the preliminaries of the mini-batch SGD, S-SGD and LSGD with full synchronization. Then, we investigate the benefits of LSGD and its limitations on the system design, which severely limits the systematic optimization space of LSGD.

\subsection{Mini-batch SGD}
Given a DL model parameterized with $w \in \mathbb{R}^d$, where $d$ is the model size. Training with a single worker is to minimize the objective function $F(w) \triangleq \mathbb{E}_{x \sim \Dcal} f(w;x)$, with sampling data $x$ from the dataset $\Dcal$. We consider that the objective function is $\beta$-smooth and $\mu$-strongly convex in this work following~\cite{Bottou2016OptimizationMF}. In each $r$-th iteration, the worker computes the gradient of a mini-batch samples $\nabla f(w_r;\xi_r)$, in which $\xi_r$ represents the sampling noise with respect to the data indexes of the current iteration $r$. The model is then updated as $w_{r+1} = w_{r} - \eta_r \nabla f(w_r;\xi_r)$~\cite{Bottou2016OptimizationMF}, where $\eta_r$ is the learning rate.

\subsection{Distributed Synchronous SGD}\label{sec:S-SGD}
We consider $K$ different workers that conduct distributed training with S-SGD~\cite{Bottou2016OptimizationMF,9488803}. In each iteration, the workers sample different data samples from the original dataset $\Dcal$. Then, the workers conduct feedforward (FP) and back-propagation (BP) to calculate the gradients $\nabla f(w_r^k;\xi_r^k)$. Next, workers communicate and average these gradients to obtain $\frac{1}{K} \sum_{k=1}^K \nabla f(w_r^k;\xi_r^k)$. The model is updated as 
\begin{equation}\label{eq:SSGD}
    w_{r+1} = w_{r} - \eta_r \frac{1}{K} \sum_{k=1}^K \nabla f(w_r^k;\xi_r^k).
\end{equation}
Communicating and averaging gradients $\nabla f(w_r^k;\xi_r^k)$ requires extra time, which limits the scalability of distributed training, especially for the low-bandwidth environments~\cite{Asteroid,ye2024galaxy,cocktailSGD,yuandecentralized,Asteroid,ye2024galaxy,ClearMLConsumerGPUs}. The S-SGD with momentum and distributed Adam also follow these process but with some extra variables to maintain. For momentum SGD~\cite{sgdm} and Adam~\cite{KingBa15}, the momentum term~\cite{sgdm} and the preconditioner~\cite{KingBa15} are updated with the averaged gradients on workers.



\subsection{Local SGD}\label{sec:local-SGD}
In LSGD, each $k$-th worker calculates the local gradient $\nabla f(w_r^k;\xi_r^k)$ and directly uses it to update the local model. After $H$ training iterations, workers communicate and average their model parameters as $\frac{1}{K} \sum_{k=1}^K w_{r}^{k}$. This process can be formulated as follows.
\begin{equation}\label{eq:localsgd}
    w_{r+1}^k =
    \begin{cases}
         w_{r}^{k} - \eta_r \nabla f(w_r^k;\xi_r^k), & \text{if} \ r+1 \% H \neq 0  \\
         \frac{1}{K} \sum_{k=1}^K (w_{r}^{k} - \eta_r \nabla f(w_r^k;\xi_r^k)), & \text{if} \ r+1 \% H = 0 
    \end{cases}.
\end{equation}
Different from S-SGD, workers with LSGD have different model parameters during local training, but synchronize them after $H$ iterations.

\textbf{Communication Overheads in Distributed Training.} Fig.~\ref{fig:comm-bp-time} shows that the communication may occupy a large amount of time compared with the computation in the DL model training and the LLM training. When bandwidth is low, training normal DL models requires larger communication time than computation~\cite{ryabinin2023swarm,tang2023fusionai,DecentMOE,cocktailSGD,yuandecentralized,10229032}. For LLMs of enormous parameters, even in a high communication bandwidth environment, the communication occupies a large ratio~\cite{gpt2,deepspeed1,ryabinin2023swarm}.

\textbf{Time Complexity.} We formulate and compare the time complexity of S-SGD and LSGD. Considering the forward and backward time is $t_{FP}$ and $t_{BP}$, and the communication time of the whole model is $t_{COMM}$, the S-SGD and LSGD require whole training time:
\begin{align}
T_{SSGD} & = R \times (t_{FP} + t_{BP} + t_{COMM}), \label{eq:TSSGD}  \\ 
T_{LSGD} & = R \times (t_{FP} + t_{BP}) + R/H \times t_{COMM}. \label{eq:TLSGD}
\end{align}
The ratio of saved time by LSGD can be calculated as $(T_{SSGD} - T_{LSGD})/T_{SSGD} = (1 - 1/H) \times t_{COMM}  /  (t_{FP} + t_{BP} + t_{COMM})$. The larger $H$ means more time reduction. However, large $H$ may cause slow convergence and worse final model performance. And the speedup of LSGD is mainly depended on the ratio of $(t_{FP} + t_{BP} )/ t_{COMM}$.

\begin{figure}[!th]
\vspace{-0.4cm}
\centering
\setlength{\abovecaptionskip}{0.cm}
\setlength{\belowcaptionskip}{-0.2cm} 
    \subfigure[Llama-2]
    {
    \includegraphics[width=0.46\linewidth]{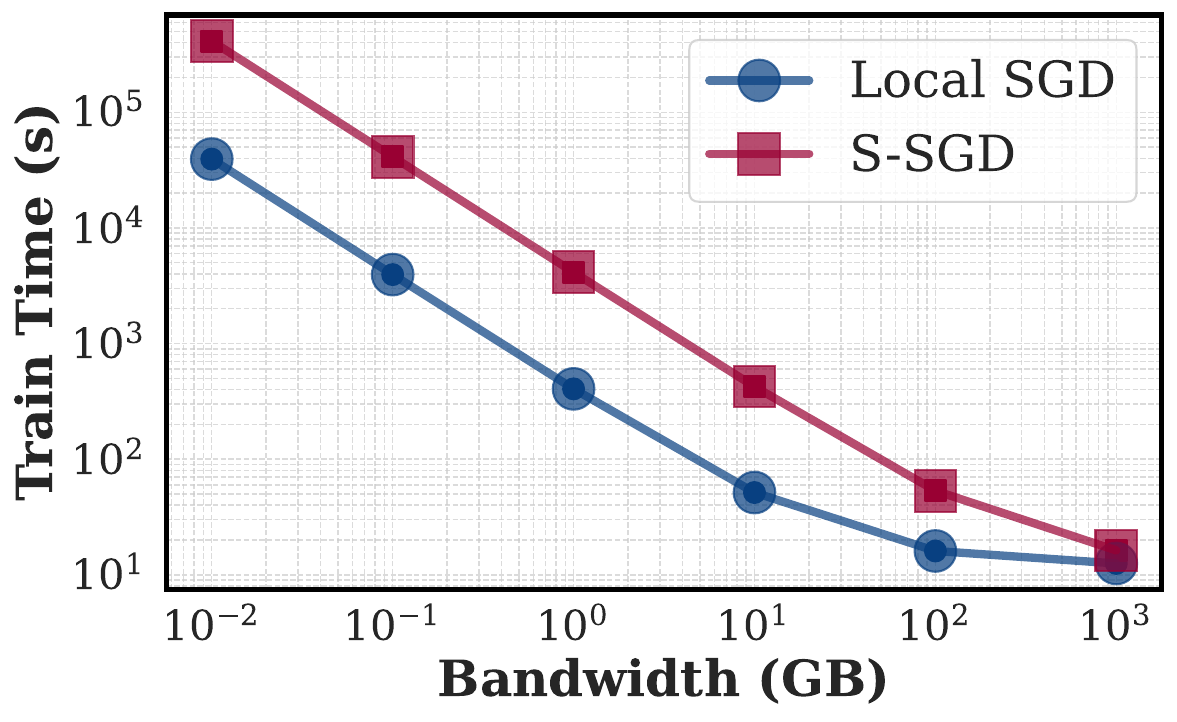}
    }
    \subfigure[GPT-2]
    {
    \includegraphics[width=0.46\linewidth]{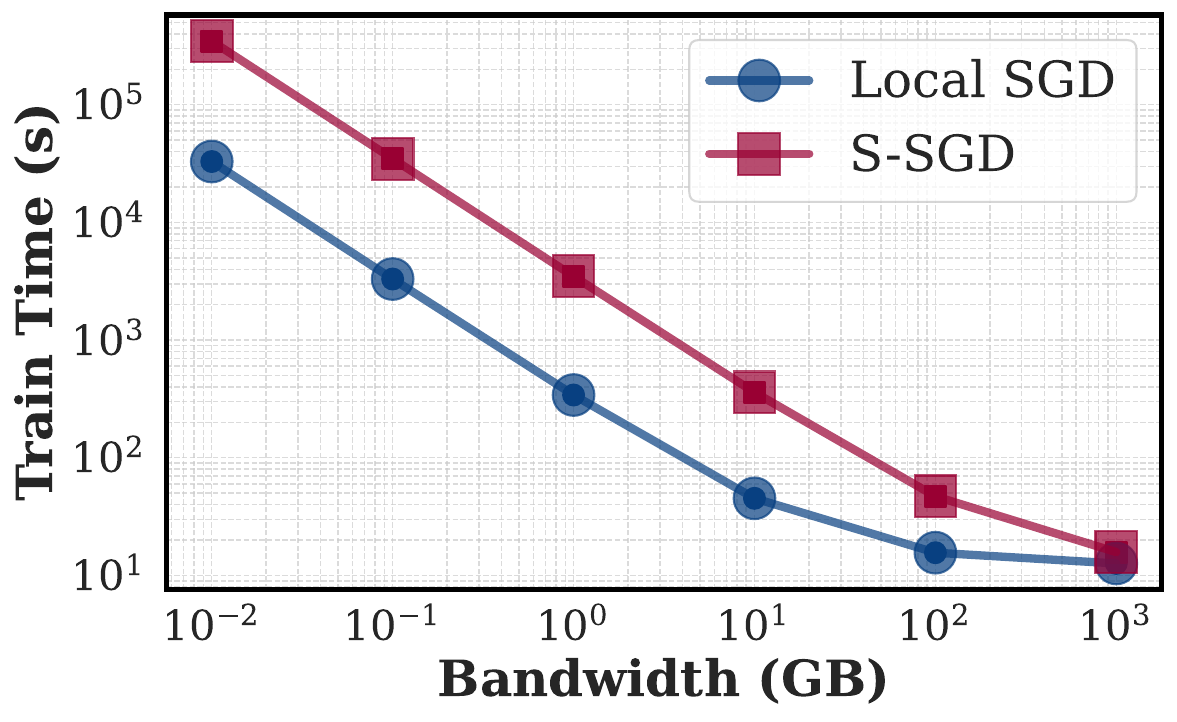}
    }
    \caption{Dissecting training time of GPT-2 and LLaMA-2 with Local SGD and S-SGD under different bandwidth.}
\label{fig:Compare-LSGD-SSGD}
\vspace{-0.1cm}
\end{figure}

\textbf{Benefits of LSGD.} The communication cost of LSGD is $O(KRS(w)/H)$, where the $S(w)$ represents the model size, $R$ the total training iterations. Compared with normal S-SGD with $O(KRS(w))$ communication cost, LSGD reduces it less than $H\times$. And this reduction is not influenced by the number of workers~\cite{shi2019distributed}. Unlike gradient compression, LSGD does not require extra computation time for compression~\cite{agarwal2022utility}. Furthermore, considering the distributed training heterogeneous GPU~\cite{jiang2020unified,guo2022hybrid}, LSGD is very suitable for implementing load balance by assigning different workers with different local training iterations, i.e. different $H_k$. Figure~\ref{fig:Compare-LSGD-SSGD} shows the reduced communication time of LSGD than S-SGD. By setting the local training iterations as $H=5$, the communication time is saved by $5\times$. Because when the communication bandwidth is low, the communication time dominates the whole iteration time. Thus, the total training time is also almost reduced for $5\times$ by LSGD.


\begin{figure*}[!ht]
\vspace{-0.0cm}
\centering
\setlength{\abovecaptionskip}{0.2cm}
\setlength{\belowcaptionskip}{0.2cm} 
    \includegraphics[width=0.95\linewidth]{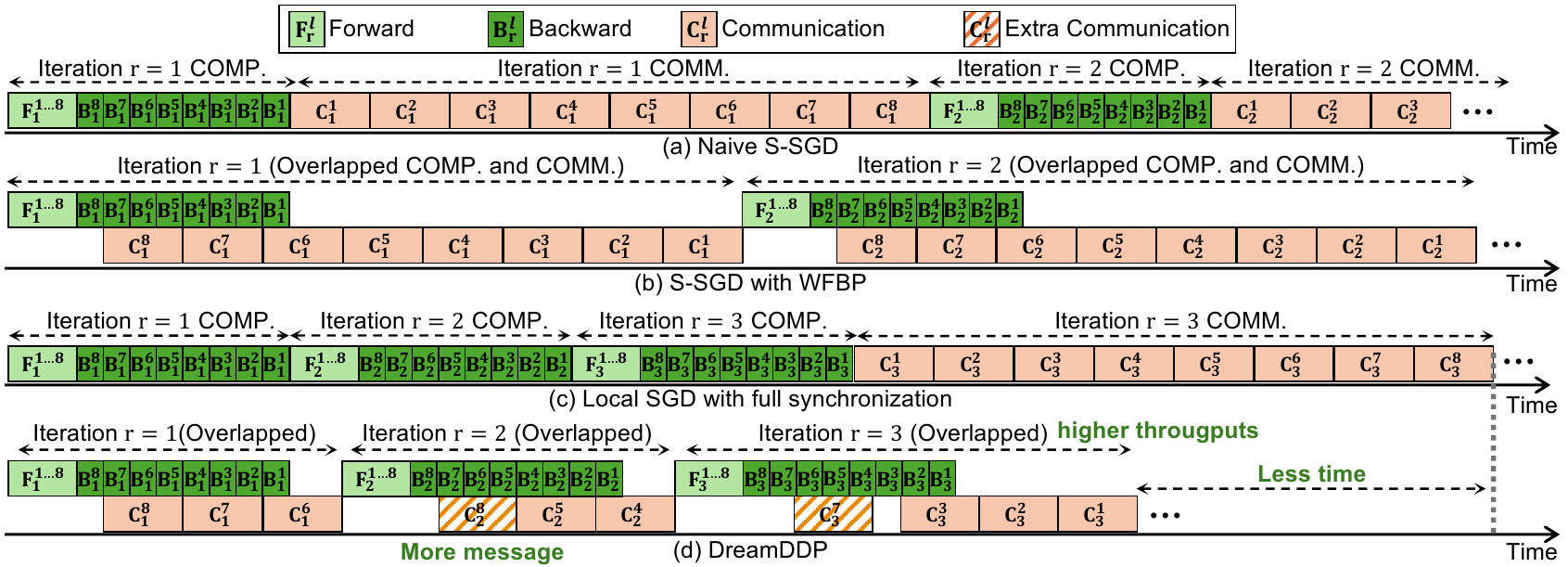}
\vspace{-0.0cm}
\caption{Timelines of different distributed training algorithms. }
\label{fig:TimelineOverall}
\end{figure*}
\vspace{-0.1cm}



\subsection{Challenges of Improving LSGD}
\textbf{Challenges 1: How to remove the hard synchronization point?}
As shown in Fig.~\ref{fig:TimelineOverall}(c), the current design of LSGD schedules BP and communication in different non-overlapping phases, of which the update operation is a \textbf{hard synchronization point}. Specifically, the update must be conducted after the gradient is obtained by BP, while the communication must be conducted after updating. Thus, the BP cannot be overlapped with the communication process, and the GPU or the communication link will stay completely idle at any given time stamp. This loses a substantial performance optimization opportunity~\cite{PytorchDistributed,wang2019adaptive} and becomes a severe intrinsic limitation of LSGD. 

In LSGD, the model divergence generated by the local training influences the convergence~\cite{stich2018localsgd,mcmahan2017communication,wang2019adaptive,mishchenko2022proxskip}. Formally, according to Eq.~\ref{eq:localsgd}, the model divergence is defined as $\Gamma_r = \frac{1}{K} \sum_{k=1}^K ||\wbar_{r} - w_{r}^{k} ||^2$, where $\wbar_{r} = \frac{1}{K} \sum_{k=1}^K w_{r}^{k}$  is the average of the model parameters. In S-SGD, the model of each worker $w_r^k$ is $\wbar_{r}$ due to the strong synchronization of gradients. Thus, we can say that $\Gamma_r$ represents the approximation error of LSGD to S-SGD. Too large $\wbar_{r}$ influences the gradient estimation. Thus, previous observations propose that the LSGD should guarantee that workers start at a synchronization point to avoid too much model divergence~\cite{mcmahan2017communication,kairouz2021advances,stich2018localsgd}. And more frequent model synchronization (less $H$) makes the LSGD more close to S-SGD, thus having better convergence speed~\cite{stich2018localsgd,wang2019adaptive,woodworth2020minibatch}.

\textbf{Challenges 2: How to overlapping parameter synchronization and computation without extra GPU memory and not influencing convergence?} 
To overlap communication and computation, we need to consider when and what to communicate. Different from the overlapping in S-SGD like WFBP~\cite{203269,9488803} where workers synchronize gradients, in LSGD, workers synchronize model parameters. Considering LLMs have enormous parameters and meet the memory wall in existing GPUs~\cite{ren2021zero,Rajbhandari2019ZeROMO}, the new method should not increase the GPU memory costs. Thus, model parameters are better to be synchronized in-place. However, communicating model parameters and averaging them in place will directly modify the model parameters. The FP and BP processes are dependent on the model parameters. Thus, in-place synchronization might influence the FP and BP, leading to incorrect feature and gradient computation.

\textbf{Challenges 3: How to schedule communication and computation?}
The scheduling of overlapping in S-SGD like WFBP~\cite{203269,9488803} is direct. All gradients in S-SGD are required to be synchronized in each training iteration. Thus, the gradient communication of each layer can be launched instantly after BP (when not considering merging small tensors~\cite{9155269}). However, in LSGD, model parameters are synchronized after $H$ iterations. For its overlapping, given $L$ layers, we need to assign communicating different layers in different iterations. The search space of this problem is extremely large as $H^L$. Such a large search space makes it difficult to find an optimal solution.

\vspace{-0.1cm}
\section{DreamDDP Design}\label{sec:DreamDDP}


DreamDDP is designed to address the above challenges to improve the training efficiency of LSGD. The overall design of DreamDDP is shown in Fig.~\ref{fig:system-design}.

Firstly, we propose \textbf{partial synchronization} based LSGD to conduct DDP training, which synchronizes different parts of the model in different iterations, offering more opportunities to overlap the communication and computation of LSGD (Section~\ref{sec:partialLSGD}). 



Secondly, DreamDDP divides the neural networks into different fine-grained $L$ layers. Built upon the partial synchronization, we analyze the dependency of the BP, FP and communication of LSGD and how to implement the overlapping with in-place partial parameter synchronization. We find that the parameter synchronization must be launched after the BP of its corresponding layer, otherwise the in-place synchronization will influence the computation of FP and BP. Thus, DreamDDP mainly overlaps the communication of one layer with the BP of its subsequent layers (Section~\ref{sec:overlap}).


Thirdly, DreamDDP integrates the profiler to profile the communication and computation time of each layer according to the given modules and GPU and bandwidth configuration. Then, we build an overall time cost model of the FP, BP and communication process with respect to $L$ layers and $H$ iterations of the partial synchronization and formulate it as an optimization problem. We identify three properties including \textit{optimal hiding}, \textit{delayed CO assignment}, \textit{at-least-one assignment} that can be leveraged to design a DFS algorithm with significantly pruned search space to solve this optimization problem (Section~\ref{sec:schedule}).

After the scheduling, we find that there may exist some idle bubbles of the communication bandwidth. Thus, we exploit them to insert more communication messages (model parameters) to fill the bubble time as long as the extra communication can be completely hidden by the computation. Thus, the synchronization frequency of these layers that fill the bubble time is higher, which accelerates the training convergence (Section~\ref{sec:bubble}).

\begin{figure}[!ht]
\vspace{-0.0cm}
\centering
\setlength{\abovecaptionskip}{0.2cm}
\setlength{\belowcaptionskip}{0.2cm} 
    \includegraphics[width=0.95\linewidth]{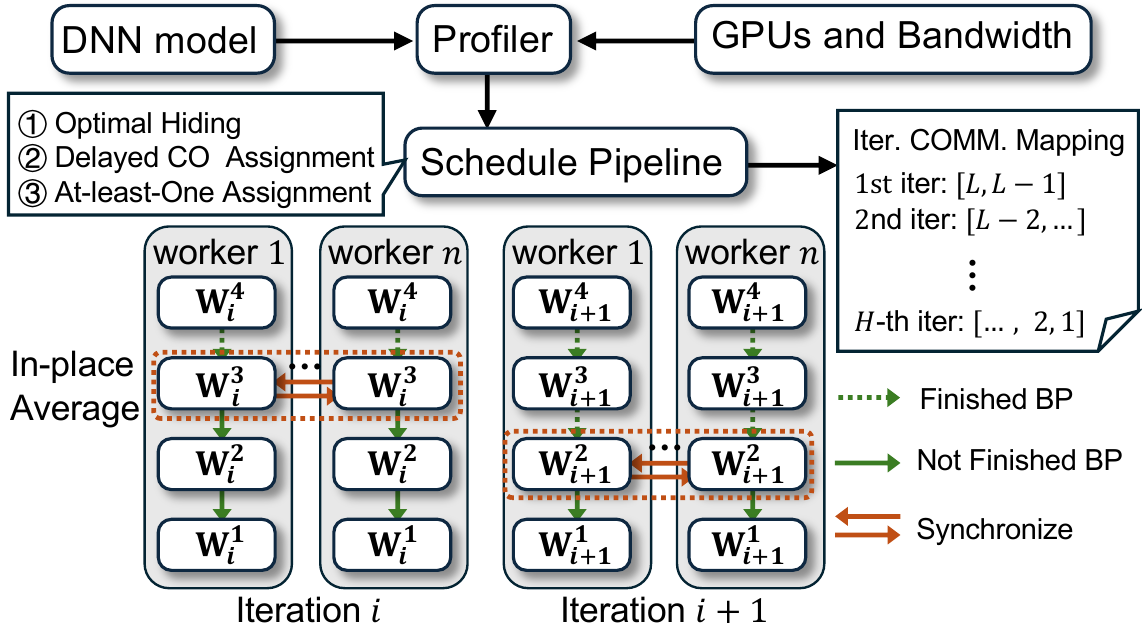}
    \caption{An overview of how DreamDDP schedules and launches the overlapping of backward propagation and parameter synchronization. In different iterations, DreamDDP synchronizes model parameters according to the scheduled iteration-layer mappings.}
    \label{fig:system-design}
\end{figure}


\subsection{Partial Synchronization: Soft Model  Synchronization Point}\label{sec:partialLSGD}

We call the model averaging in the original LSGD with all parameters as the full synchronization as shown in Fig.~\ref{fig:TimelineOverall}(a). The strict synchronization of LSGD makes the training time of $H$ iterations cannot be utilized to overlap the communication like WFBP does (Fig.~\ref{fig:TimelineOverall}(b)).

\begin{algorithm}[!t]
\caption{LSGD with partial synchronization}\label{algo:partial-LSGD}
\small
\textbf{Input: } Initialized model $w_0$, dataset $\Dcal$, workers $\{ 1,...,K\}$, total iteration $R$, learning rate $\eta$, synchronization frequency $H$. \\
\textbf{Output: } Final trained model $w_R$. 
\begin{algorithmic}[1]
    \For{$r =1,...,R$}
    \For{worker $k \in \{ 1,...,K\}$ in parallel}
        \State $\nabla f(w_r^k;\xi_r^k) \leftarrow$ FP and BP; \Comment{Obtain gradients}
        \State $w_{r+1/2}^k = w_{r}^{k,l} - \eta_r \nabla^l f(w_r^k;\xi_r^k)$; \Comment{Update model}
        \For{layer $l\in \Lcal$}
        \If{$r+1 \% H = H_l$}   \Comment{Partial synchronization}
            \State $w_{r+1}^{k,l} = \frac{1}{K} \sum_{k\in \Mcal} w_{r+1/2}^k $;
        \Else
            \State $w_{r+1}^{k,l} = w_{r+1/2}^k $;
        \EndIf            
    \EndFor
    \EndFor
    \EndFor
    \State Return $w_{R}$;
\end{algorithmic}
\end{algorithm}

To hide more communication of LSGD within computation, we propose partial synchronization. Specifically, we split the whole model into different segments $w_l$ according to the set of layer indexes  $\Lcal$ and uniformly assign them to different local training iterations for communicating as shown in Fig.~\ref{fig:TimelineOverall}(d). Now, compared with Fig.~\ref{fig:TimelineOverall}(c), partial synchronization provides more opportunity to overlap communicating different layers within the computation.


Now, separated layers can be synchronized within the computation time. By synchronizing layers after their backward propagation, the overlapping is implemented with guaranteeing parameter consistency. Besides, there is no extra GPU memory occupation. Formally, we give the process of LSGD with partial synchronization in Algorithm~\ref{algo:partial-LSGD} and the following definition.

\begin{definition}\label{def:partial}
(\textbf{Partial Synchronization.}) Assuming that the layer indexes $\Lcal=\{ 1, ...,L\}$ of the model are sequentially split to $H$ disjoint subsets $\Lcal_{1:H}=\{ \Lcal_1, ...,\Lcal_H \}$. Given the synchronization period $H $, the updating form of LSGD with partial synchronization is 
\begin{small}
    \begin{equation}\label{eq:partial-LSGD}
    w_{r+1}^{k,l} =
    \begin{cases}
         w_{r}^{k,l} - \eta_r \nabla^l f(w_r^k;\xi_r^k), & \text{if} \ r+1 \% H \neq H_l  \\
         \frac{1}{K} \sum_{k=1}^K (w_{r}^{k,l} - \eta_r \nabla^l f(w_r^k;\xi_r^k)), & \text{if} \ r+1 \% H =  H_l
    \end{cases},
\end{equation}
\end{small}
in which $l\in\Lcal$, and $H_l=\Ical(l)$ is the index function that indicates which subset of $\Lcal_{1:H}$ the $l$-th layer belongs to. Note that the gradient is calculated with respect to the whole model $w_r^k$. We use $ \nabla^l f(w_r^k;\xi_r^k)$ to represent the gradient of the $l$-th layer.
\end{definition}
\begin{example}\label{ex:ENP}
(\textbf{Equal-Number Partition.}) Given $H=2/L$, we have a partition as of $\Lcal$ as $\Lcal_1 = \{1,2\}, \Lcal_2 = \{3,4\},...,\Lcal_H = \{L-1,L\}$, which equally assigns $L$ layers to $H$ sets.  We write $\Lcal_{1:H} = \{ \Lcal_1, \Lcal_2,...,\Lcal_H\}$  for simplicity. The Definition~\ref{def:partial} and Algorithm~\ref{algo:partial-LSGD} show that the averaging periods of layer sets are different. Similar to full synchronization, within $H$ iterations, model parameters are all synchronized once. In other words, the accumulated divergence $\Gamma_r$ in full synchronization is cleared after $H$ iterations, while the layer-wise $\Gamma_r^l$ is cleared after $H$ iterations in partial synchronization. 
\end{example}

Note that the whole model divergence $\Gamma_r$ always exists during the entire training process. This is somehow contrary to full synchronization~\cite{kairouz2021advances,stich2018localsgd} that requires to completely clear the model divergence. However, we empirically show that LSGD with partial synchronization has less model divergence $\Gamma_r$, which benefits convergence. Fig.~\ref{fig:PSvsFull} shows using LSGD with partial and full synchronization to train ResNet-18 with 32 workers. Fig.~\ref{fig:PSvsFull-Convergence} shows that the convergence speed of full synchronization is clearly slower than partial synchronization. And the model divergence of full synchronization is periodically accumulated and cleared, while the partial synchronization maintains the divergence at a lower degree.


\textbf{Reduced Divergence Amplification.} Intuitively, we  show that the larger degree of model divergence in full synchronization comes from the layer-wise divergence amplification. Considering that intermediate features $z^{k,l}(x) = f^{k,l}(w^{k,1:l},x)$ on client $k$, where $f^{k,l}$ represents the mapping $x \to z$ through the former $1,...,l$ layers, the larger divergence between $w^{k,1:l}$ introduces larger divergence between $z^{k,l}$. And the subsequent layer features are dependent on the former intermediate features, which implies that the divergence will be amplified layer by layer. Thus, during the late stage of one period in full synchronization, the total divergence is amplified to a large degree, as shown in Fig.~\ref{fig:PSvsFull-Divergence}. However, the partial synchronization frequently eliminates the divergence of some layers, which alleviates the layer-wise divergence amplification effect.

\begin{figure}[!ht]
\vspace{-0.0cm}
\centering
\setlength{\abovecaptionskip}{0.2cm}
\setlength{\belowcaptionskip}{0.2cm} 
    \subfigure[Test accuracy]
    {
    \includegraphics[width=0.46\linewidth]{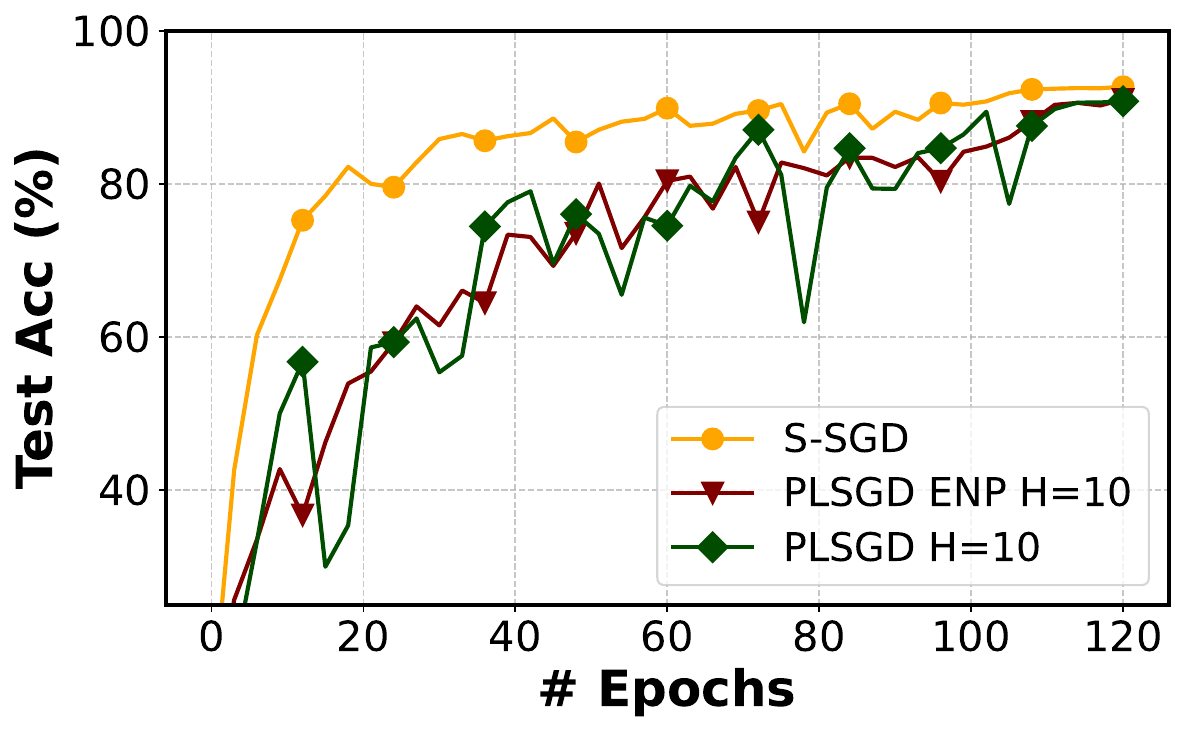}\label{fig:PSvsFull-Convergence}
    }
    \subfigure[Model divergence]
    {
    \includegraphics[width=0.46\linewidth]{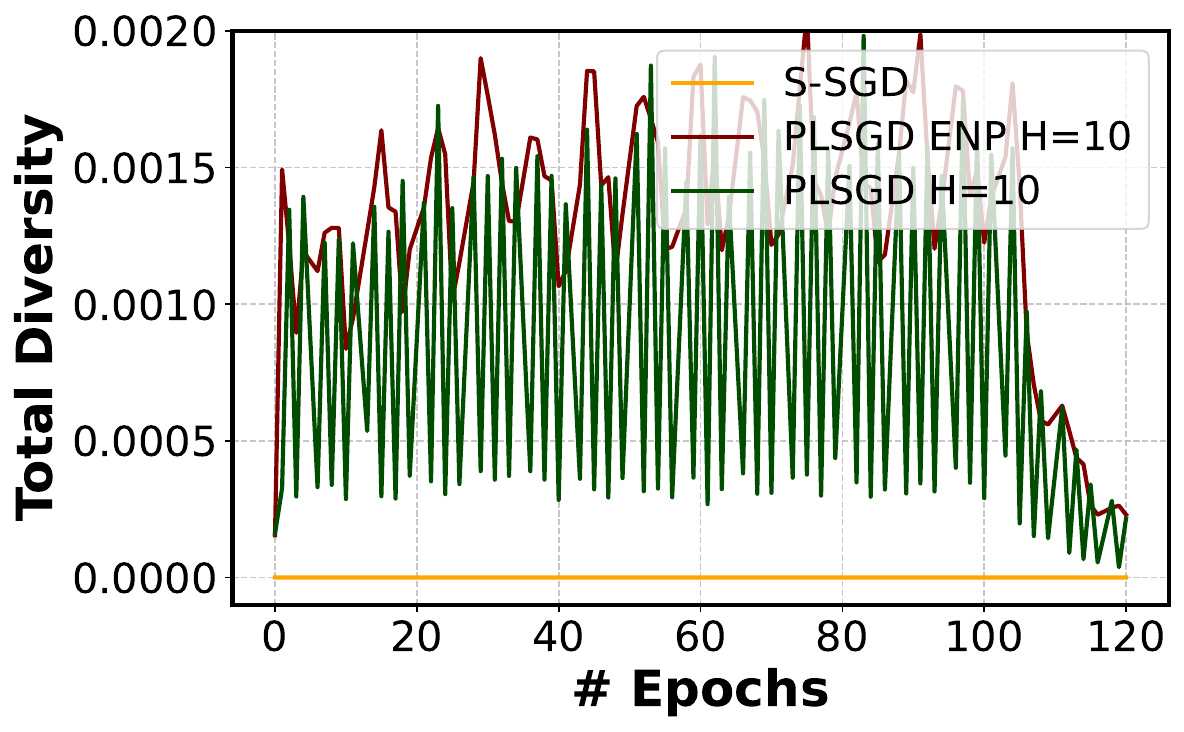}\label{fig:PSvsFull-Divergence}
    }
\vspace{-0.0cm}
\caption{Training ResNet-18 on 32 workers.}
\label{fig:PSvsFull}
\vspace{-0.1cm}
\end{figure}



Now, we theoretically prove that the LSGD with partial synchronization has the same convergence rate as the S-SGD~\cite{Bottou2016OptimizationMF}.

\begin{theorem}
\label{theo:converge-bound}
    Let $f$ be $\beta$-smooth and $\mu$-strongly convex, $\mathbb{E}_{\xi_r^k}||\nabla f(w_r^k;\xi_r^k)-\nabla f(w_r^k)||^2\leq \sigma^2$, $\mathbb{E}_{\xi_r^k}||\nabla f(w_r^k;\xi_r^k)||^2\leq G^2$, for $r=0,1,...,R-1$, where $\{w_r^k\}_{r=0}^R$ for $k\in[K]$ are generated according to (\ref{eq:partial-LSGD}) with $H_l\leq H, l=1,2,...,L$ and for stepsizes $\eta_r=\frac{4}{\mu(a+r)}$ with shift parameter $a>\max\{16\kappa,H\}$, for $\kappa=\frac{L}{\mu}$. Then
    \begin{equation}
    \small
    \begin{split}
        \mathbb{E}f(\hat{w}_R)-f^*&\leq \frac{\mu a^3}{2S_R}||w_0-w^*||^2\\
        &+\frac{4R(R+2a)}{\mu K S_R}\sigma^2+\frac{256R}{\mu^2S_R}G^2H^2\beta
    \end{split}
    \end{equation}
    where $\hat{w}_R=\frac{1}{KS_R}\sum_{k=1}^K\sum_{r=0}^{R-1}p_rw_r^k$, for $p_r=(a+r)^2$ and $S_R=\sum_{r=0}^{R-1}p_r\geq\frac{1}{3}R^3$.
\end{theorem}
\noindent \textbf{Proof Sketch.} We follow~\cite{Bottou2016OptimizationMF} to assume the gradient sampling variance is bounded as $\mathbb{E}_{\xi_r^k}||\nabla f(w_r^k;\xi_r^k)-\nabla f(w_r^k)||^2\leq\sigma^2$ for any worker $k\in[K]$ and iteration $ r\in[R]$, which is a general assumption in analysis of convergence. Thus, the distributed training with multiple workers helps reduce the gradient variance as  $\mathbb{E}||g_r-\Bar{g}_r||^2\leq\frac{\sigma^2}{K}$. And the gradient magnitude is upper bounded as  $\mathbb{E}||\nabla f(w_r^k;\xi_r^k)||^2\leq G^2$. Then, we define a virtual state sequence $\Bar{w}_r=\frac{1}{K}\sum_{k=1}^Kw_r^k$ and gradient sequences  $g_r=\frac{1}{K}\sum_{k=1}^K\nabla f(w_r^k;\xi_r^k), \Bar{g}_r=\frac{1}{K}\sum_{k=1}^K\nabla f(w_r^k)$. The detailed proof is provided in Appendix~\ref{apdx:proof}.

\begin{corollary}\label{coro:convergence}
    Let $\hat{w}_R$ be defined as in Theorem~\ref{theo:converge-bound}, for parameter $a=\max\{16\kappa,H\}$. Then
    \small
    \begin{equation*}
    \small
    \Ebb f(\hat{w}_R)-f^*=\mathcal{O}\Big(\frac{\kappa H^2}{\mu R^2}+\frac{\kappa^3+H^3}{\mu R^3}\Big)G^2 +\mathcal{O}\Big(\frac{1}{\mu K R}+\frac{\kappa+H}{\mu KR^2}\Big)\sigma^2.
    \end{equation*}
\end{corollary}

\noindent \textbf{Remark.} Theorem~\ref{theo:converge-bound} and Corollary~\ref{coro:convergence} show that the LSGD with partial synchronization has the same convergence rate as S-SGD~\cite{Bottou2016OptimizationMF,woodworth2020minibatch} as $\mathcal{O}(1/R)$. Small $H$ can help accelerate the convergence of other terms containing $\mathcal{O}(1/R^2)$ or $\mathcal{O}(1/R^3)$, which are not dominant in the convergence bound with respect to the training iterations $R$.




\textbf{Wall-clock Time Complexity Analysis.} Following definitions in Section~\ref{sec:prelimilary}, we analyze and compare the time complexity of one-iteration overlapped LSGD and partial synchronization. With Definition~\ref{def:partial}, the new training time of LSGD with partial synchronization is
\begin{equation}\label{eq:T-PartLSGD}
\begin{split}
    T_{P-LSGD} = & R \times t_{FP} +  \frac{R}{H} \sum_{h=1}^{H}(t_{BP}^{\Lcal_{1:h-1}} + t_{BP}^{h_0} \\
    & + \text{max}(t_{BP}^{\Lcal_{h:H}}-t_{BP}^{h_0}, t_{COMM}^{\Lcal_h})),
\end{split}
\end{equation}
in which, $h_0$ represents the first layer in set $\Lcal_h$, $t_{BP}^{\Lcal_{1:h-1}}=\sum_{i=1}^{h-1} (\sum_{l\in \Lcal_i} t_{BP}^l)$ is the backward time of layers in $\Lcal_1,...,\Lcal_h$,  $t_{BP}^{\Lcal_{h:H}}=\sum_{i=h}^H (\sum_{l\in \Lcal_i} t_{BP}^l)$ is the backward time of layers in $\Lcal_{h+1},...,\Lcal_H$, $t_{COMM}^{\Lcal_i}=\tau_{COMM}^{h_{| \Lcal_i |-1}} + t_{COMM}^{h_{| \Lcal_i |-1}} - \tau_{COMM}^{h_0}$ is the communication time of layers in $\Lcal_{i}$, where $\tau_{COMM}^{h_j} = \max(t_{BP}^{h_{0:j}}, \tau_{COMM}^{h_{j-1}} + t_{COMM}^{h_{j-1}})$ defines the timestamp for the communication time of layer ${h_j}$. The expression $t_{BP}^{h_{0:j}}$ indicates the completion timestamp for the backward time of layer ${h_j}$, and $\tau_{COMM}^{h_{j-1}} + t_{COMM}^{h_{j-1}}$ marks the completion timestamp for the communication time of layer $h_{i-1}$. The term $t_{BP}^{\Lcal_{h:H}}-t_{BP}^{h_0}$ comes from that the communication of the first layer in the set $\Lcal_h$ must wait for the backward to avoid simultaneous parameter accessing.






\subsection{Overlapping Communication and Computation}\label{sec:overlap}

\begin{figure}[!ht]
\vspace{-0.0cm}
\centering
\setlength{\abovecaptionskip}{0.2cm}
\setlength{\belowcaptionskip}{0.2cm} 
\includegraphics[width=0.95\linewidth]{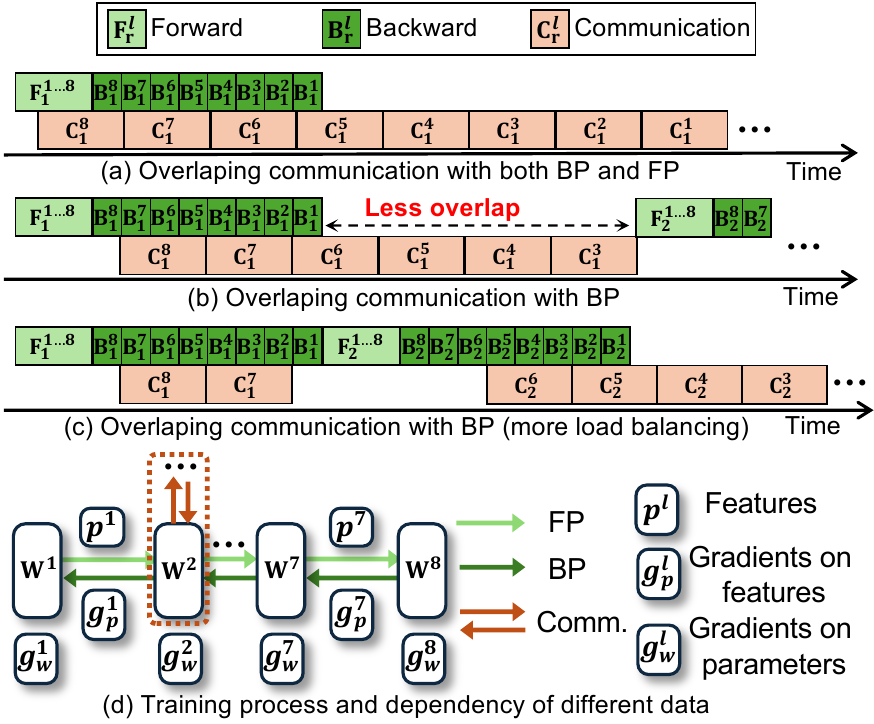}
\caption{The disection of the dependency and overlap opportunity between communication and computation of LSGD.}
\label{fig:overlapping}
\end{figure}

With the partial synchronization, Theorem~\ref{theo:converge-bound} shows that the convergence rate of partial synchronization-based LSGD (PLSGD) has the same convergence rate as LSGD. However, the analysis and Theorem~\ref{theo:converge-bound} in Section~\ref{sec:partialLSGD} do not provide the details of how to overlap communication and computation. In this section, we disect the dependency between communication and computation of LSGD. With the principle of introducing no extra GPU memory occupation, we propose to overlap the communication of layers with the BP of its corresponding layer. 

Figure~\ref{fig:overlapping} (d) shows the FP, BP and the parameter synchronization of LSGD. In LSGD, workers synchronize and average the model parameters $W$ instead of the gradients on the parameters $g_W$ in S-SGD. The FP and BP computation of each layer is dependent on the layer parameters $W_l$ and $g_{W_l}$. Thus, in-place overlapping communication of layers with the FP computation is not feasible because this will lead to incorrect FP and BP computation. Thus, while overlapping communication with the FP computation helps to hide more communication time as shown in Figure~\ref{fig:overlapping} (a), it is not feasible considering that we need in-place averaging of parameters $W$ to avoid extra memory occupation. Besides, considering that the FP normally occupies less time that BP, the communication time is not dominated by the FP time. 

Given $L$ layers that are required to be allocated to $H$ iterations for synchronization, we need to consider how to schedule the overlapping to maximally reduce the time cost. Figure~\ref{fig:overlapping} (b) (c) shows that directly communicating all layers in the first local iteration may not be optimal. When the communication time is much larger than the BP time, a more load balancing overlaping strategy like shown in Figure~\ref{fig:overlapping} (c) is better. In the next section, we formally define the overall training time of PLSGD within one communication period as the optimization goal, and identify three properties that can be used to reduce the search space of the solution.

\subsection{Scheduling and Optimization}\label{sec:schedule}

With partial synchronization, the intuitive way to reduce the training time is minimizing the training time $T_{P-LSGD}$ by adjusting the partition of layer index sets $\Lcal_{1:H}$. According to the time cost of $T_{P-LSGD}$ (Eq.~\ref{eq:T-PartLSGD}), the time $t_{FP}$, $t_{BP}$ and $t_{COMM}$ are fixed, and different partitions of layer indexes in to $\Lcal_{1:H}$ will influence $ \sum_{h=1}^{H}\text{max}(t_{BP}^{\Lcal_{h:H}}-t_{BP}^{h_0}, t_{COMM}^{\Lcal_h})$. If we can adjust the $\Lcal_{1:H}$ well,  $ t_{COMM}^{\Lcal_h}$ might be fully overlapped. The scheduling problem is highly complex. To address this, we treat each layer partition as an interval, which allows us to approximate a solution. Note that the communication of the last layer cannot be overlapped, because the communication must wait for the backward. Then, we propose the following objective function $P_H^L$ as the goal of DreamDDP.
\begin{small}
    
\begin{equation}\label{eq:minPHL}
    \min_{\Lcal_{1:H}^{L:1}} P_H^L = \sum_{h=1}^{H} (t_{BP}^{\Lcal_{1:h-1}} + t_{BP}^{h_0} + \text{max}(t_{BP}^{\Lcal_{h:H}}-t_{BP}^{h_0}, t_{COMM}^{\Lcal_h})),
\end{equation}
\end{small}
which is a combinatorial optimization problem, and $\Lcal_{1:H}^{L:1}$ represents the possible partitions $\Lcal_{1:H}$ of layers $\{L,...,1\}$. The brute force solution of this problem requires time complexity of $O(C_{H+L}^L)=\frac{(H+L)!}{L!H!}$.


\textbf{Searching Process and Sub-structure Problem.} To efficiently solve this problem, we identify the optimal substructure of problem~\ref{eq:minPHL}.  Now, we solve this problem with an efficient searching algorithm with reduced complexity as $\Ocal(2^{\text{min}(L-H, H)})$. The search space is significantly reduced by the optimal substructure of this problem. The searching process is to iteratively decide whether to assign $l$-th layer to the $h$-th set $\Lcal_h$. From this perspective, given assigned layers $\{L,...,l+1 \}$ in disjoint sets $\{ \Lcal_1, ...,\Lcal_{h-1}, \Lcal_{h}\}$, we define the sub-problem as follows,
\begin{small}
\begin{equation}\label{eq:subminPHL}
\min_{\Lcal_{h:H}^{l:1}} S_{h}^l = \sum_{i=h}^{H}( t_{BP}^{\Lcal_{1:i-1}} + t_{BP}^{i_0} +\text{max}(t_{BP}^{\Lcal_{i:H}}-t_{BP}^{i_0}, t_{COMM}^{\Lcal_{i}})),
\end{equation}
\end{small}
in which the $i$ iterates from $h$ to $H$, $H$ and $L$ are fixed in $S_h^l$, optimizing on $\Lcal_{h:H}^l$ means to split layers $\{1,...,l\}$  to sets $\{\Lcal_h, \Lcal_{h+1},..., \Lcal_H\}$. Note that $l$ follows descending order as the backward process is from output to input, and $h$ follows the ascending order as the local training follows time order in a synchronization period ${r+1,...,t+H}$. When, $h=1$ and $l=L$ in $S_h^l$,  $S_1^L$ becomes the original problem~\ref{eq:minPHL}, i.e. $P_H^L =  S_{1}^L$. 
\begin{definition}\label{def:assign}
    (\textbf{Assignment $A^l$.}) We call assigning a layer $l$ into $\Lcal_h $ as replacing the set $\Lcal_h$ by $\Lcal_h \cup\{l\} $ in original sets $\Lcal_{1:H}$ and we have $A^l=h$.
\end{definition}



Optimizing $ P_H^L$ can be seen as an iterative searching process from optimizing $S_1^L$ to $S_H^1$. In each iteration of optimizing $S_h^l$,  we have $\{L,...,l+1 \}$ that have been assigned in disjoint sets $\{ \Lcal_1, ...,\Lcal_{h-1}, \Lcal_{h}\}$. Then, based on whether to assign the layer $l$ to set $\Lcal_h$, each $S_h^l$ has the substructure as:




\begin{equation}\label{eq:substructure}
S_{h}^l =
\begin{cases}
    \begin{aligned}
            S_{h}^{l-1}  & +\text{max}(t_{BP}^{\Lcal_{h:H}}-t_{BP}^{h_0}, t_{COMM}^{\Lcal_h\cup\{l\}}) \\
            &  - \text{max}(t_{BP}^{\Lcal_{h:H}}-t_{BP}^{h_0}, t_{COMM}^{\Lcal_h})
\end{aligned}, 
        & \text{if} \ A^l=h  \\
    \begin{aligned}
        S_{h+1}^l  & +   t_{BP}^{\Lcal_{1:h-1}} + t_{BP}^{h_0} 
\\ & + \text{max}(t_{BP}^{\Lcal_{h:H}}-t_{BP}^{h_0}, t_{COMM}^{\Lcal_h}), 
\end{aligned}
& \text{if} \ A^l \neq h
\end{cases},
\end{equation}

in which, for simplicity, we do not expand terms $t_{BP}^{\Lcal_{h:H}}-t_{BP}^{h_0}=\sum_{i=h}^H (\sum_{l\in \Lcal_i} t_{BP}^l)-t_{BP}^{h_0}$ and $ t_{COMM}^{\Lcal_h}=\sum_{l\in \Lcal_h} t_{COMM}^l$ that follows definition in. The first term means that we assign $l$ to $\Lcal_h$, thus the total time cost $S_{h}^{l-1}$ re-estimates the time cost of the set $\Lcal_h$ for the $h$-th training round. Note that the $t_{BP}^{\Lcal_{h:H}} = \sum_{i=h}^H (\sum_{l\in \Lcal_i} t_{BP}^l)$  actually is fixed when scheduling $S_h^l$ no matter whether $l\in \Lcal_h$, because all left layers $\{\Lcal_h,...,\Lcal_{H} \}$ will conduct BP operations at $h$-th iteration, which is fixed for $S_h^l$. We write the optimal solution of $\min_{\Lcal_{h:H}^{l:1}} S_{h}^l$ as $S_{h,\star}^l$.

The optimal solution of problem~\ref{eq:minPHL} $P_H^L$ and $S_1^L$ should minimize the communication time as much as possible, because the total BP time is fixed as $t_{BP}^{\Lcal}$ for each iteration $h$ from $1$ to $H$. The minimum of $P_H^L$ should contain the minimal value of communication. Now, we give following definition and properties which are used to analyze and reduce the search space.

\begin{definition}\label{def:CHA}
    (\textbf{Communication-Hide Assignment.}) Given assigned layers $\{L,...,l+1 \}$ in disjoint sets $\{ \Lcal_1, ...,\Lcal_{h}\}$, and the optimization goal $ S_{h}^{l}$, we call the assignment (Definition~\ref{def:assign})  $A^l=h$  is Communication-Hide (CH) if it satisfies  
\begin{equation}\label{eq:hided-assign}
    t_{BP}^{\Lcal_{h:H}}-t_{BP}^{h_0} \geq  t_{COMM}^{\Lcal_h\cup\{l\}}.
\end{equation}
Otherwise, we call  the assignment  $A^l=h$  is communication-overflowed (CO). 
\end{definition}



\begin{figure}[!ht]
\vspace{-0.0cm}
\centering
\setlength{\abovecaptionskip}{0.2cm}
\setlength{\belowcaptionskip}{0.2cm} 
    \subfigure[At iteration $r$, $ S_{h}^{l} \to  S_{h}^{l-1}$ ($A^l=h$).]
    {
    \includegraphics[width=0.98\linewidth]{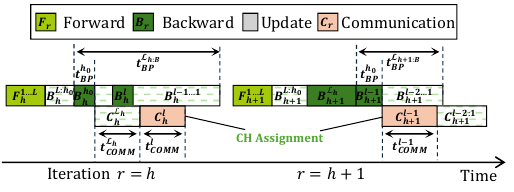}\label{fig:GreedyAssign-yes}
    }
    \subfigure[At iteration $r$, $ S_{h}^{l} \to  S_{h+1}^{l}$($A^l\neq h$).]
    {
    \includegraphics[width=0.98\linewidth]{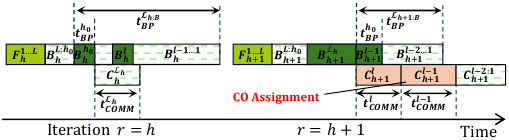}\label{fig:GreedyAssign-no}
    }
    \caption{Comparing whether assigning $A^l=h$ when $A^l_h$ is CH at iteration $r$. The time of $F_{h}^{1...L}$, $B_{h}^{L...l}$, $B_{h}^{l-1...1}$ are shrunk in Figure for clear illustration.}
\label{fig:GreedyAssign}
\vspace{-0.1cm}
\end{figure}






\begin{property}\label{prop:OptimalHiding}
(\textbf{Optimal Hiding.}) 
Considering the case that $A^l=h$ is CH for $S_h^l$, then $\text{max}(t_{BP}^{\Lcal_{h:H}}-t_{BP}^{h_0}, t_{COMM}^{\Lcal_h\cup\{l\}})= \text{max}(t_{BP}^{\Lcal_{h:H}}-t_{BP}^{h_0}, t_{COMM}^{\Lcal_h}) = t_{BP}^{\Lcal_{h:H}}-t_{BP}^{h_0}$, we have $S_{h,\star}^{l-1}\leq S_{h+1,\star}^{l} + t_{BP}^{\Lcal}$. The intuitive interpretation is shown in Fig.~\ref{fig:GreedyAssign}. If $l$-th communication is assigned at iteration $h$ and hidden by the computation, it is easier for unscheduled $H-h$ iterations $\{h+1,...,H\}$ to hide the communication of other $l-1$ layers $\{l-1,...,1\}$ than $l$ layers. Delaying the communication of $l$-th layer might lead to failed overlap of subsequent communications.

Thus, based on the optimal hiding (property~\ref{prop:OptimalHiding}), DreamDDP assigns $A^l=h$ if it is fully hidden by computation (line 20 in Algorithm~\ref{algo:DreamDDP}). This semi-greedy scheduling significantly reduces the search space. 
\end{property}




\begin{algorithm}[!t]
\caption{Assigning Communication of DreamDDP}\label{algo:DreamDDP}
\small
\textbf{Input: } Layer set $\Lcal$, synchronization frequency $H$, backward times $\{t_{BP}^l |l\in \Lcal \}$, communication times $\{t_{COMM}^l |l\in \Lcal \}$, initial assignment set $\Psi = \{\Lcal_h= \emptyset| h \in \{ 1,...,H\}\}$, solutions set $\Omega$. \\
\textbf{Output: } Final trained model $w_R$. 
\begin{algorithmic}[1]
    \State \Call{SolveSubProblem}{$\Psi$, $L$, $L$, $1$, $H$} \Comment{Scheduling}
    \State $\Psi^\star  = \argmin_{\Psi \in  \Omega} P_H^L$; \Comment{Find one optimal solution}
    \For{$\Lcal_h  \in \Psi^\star$} 
        \State $ l_h^{-} = \min_{l} t_{COMM}^{L...l} \ s.t. \ l \notin \Lcal_h    $ and Eq.~\ref{eq:morelayers}.  
        \State $ \Lcal_h = \Lcal_h \cup \{L,...,l \}   $ \Comment{Add more communication}
    \EndFor
\State Return $\Psi^\star$.
\State
\Procedure{SolveSubProblem}{$\Psi$, $l$, $L$, $h$, $H$}
    \If{$l == 0$}
       \State Add $\Psi$ in $ \Omega$. \Comment{Get one solution}
       \State Return
    \EndIf
    \If{$h==H$}
        \For{$i =l,...,1$}        
            \State Replace $\Lcal_h \in \Psi$ as $\Lcal_h=\Lcal_h \cup\{i\} $
            \State Add $\Psi$ in $ \Omega$. \Comment{Get one solution}
        \EndFor
        \State Return
    \EndIf

    \If{$\Lcal_h$ is $\emptyset$}  \Comment{At-least-one Assignment}
        \State Replace $\Lcal_h \in \Psi$ as $\Lcal_h=\Lcal_h \cup\{l\} $
        \State \Call{SolveSubProblem}{$\Psi$, $l-1$, $L$, $h$, $H$}
    \ElsIf{$  t_{BP}^{\Lcal_{h:H}}-t_{BP}^{h_0} \geq  t_{COMM}^{\Lcal_h\cup\{l\}}  $}   \Comment{Optimal Hiding}
        \State Replace $\Lcal_h \in \Psi$ as $\Lcal_h=\Lcal_h \cup\{l\} $
        \State \Call{SolveSubProblem}{$\Psi$, $l-1$, $L$, $h$, $H$}
    \ElsIf{$ t_{BP}^{\Lcal_{h:H}}-t_{BP}^{h_0} <  t_{COMM}^{\Lcal_h}$}    \Comment{Delayed COA}
        \State \Call{SolveSubProblem}{$\Psi$, $l$, $L$, $h+1$, $H$}
    \Else                   \Comment{ Search}
        \State $\Psi_{-} =$ DeepCopy($\Psi$)
        \State Replace $\Lcal_h \in \Psi_{-}$ as $\Lcal_h=\Lcal_h \cup\{l\} $ \Comment{Assign $A^l=h$}
        \State \Call{SolveSubProblem}{$\Psi_{+}$, $l-1$, $L$, $h$, $H$}   
        \State $\Psi_{+} =$ DeepCopy($\Psi$)
        \State \Call{SolveSubProblem}{$\Psi_{+}$, $l$, $L$, $h-1$, $H$}  
    \EndIf
\EndProcedure
\end{algorithmic}
\end{algorithm}

\begin{figure}[!ht]
\vspace{-0.0cm}
\centering
\setlength{\abovecaptionskip}{0.2cm}
\setlength{\belowcaptionskip}{0.2cm} 
    \subfigure[At iteration $r$, $ S_{h}^{l} \to  S_{h+1}^{l}$($A^l\neq h$).]
    {
    \includegraphics[width=0.98\linewidth]{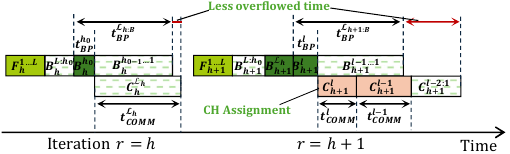}\label{fig:Overflow-CHA}
    }
    \subfigure[At iteration $r$, $ S_{h}^{l} \to  S_{h}^{l-1}$ ($A^l=h$).]
    {
    \includegraphics[width=0.98\linewidth]{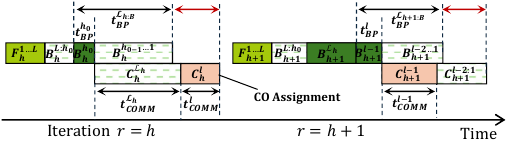}\label{fig:Overflow-COA}
    }
    \caption{Comparing whether assigning $A^l=h$ when the $t_{COMM}^l$ is completely not hided at iteration $r$.}
\label{fig:Overflow}
\vspace{-0.1cm}
\end{figure}

\begin{property}\label{prop:DelayedCO}
(\textbf{Delayed CO Assignment.}) Considering the case that communication of set $\Lcal_h$ cannot be hided by the BP computation, i.e. $t_{BP}^{\Lcal_{h:H}}-t_{BP}^{h_0} \leq t_{COMM}^{\Lcal_h} $. In this situation, the assignment $A^l=h$ is also CO, because $t_{BP}^{\Lcal_{h:H}}-t_{BP}^{h_0} \leq t_{COMM}^{\Lcal_h} \leq t_{COMM}^{\Lcal_h\cup\{ l\}} $ and Definition~\ref{def:CHA}. Now, $A^l=h$ will lead to extra communication time as $ t_{COMM}^{l} $, which is completely not hided and increases the total training time as shown in Fig.~\ref{fig:Overflow-CHA}. However, if we delay communicating layer $l$ at iteration $h+1$, the communication cost might be fully or partly hidden by the computation as shown in Fig.~\ref{fig:Overflow-COA}. 

Note that the delayed CO assignment (Property~\ref{prop:DelayedCO}) may not be optimal. However, this heuristic insight can help us significantly reduce the search space. For each $h$, if $t_{BP}^{\Lcal_{h:H}}-t_{BP}^{h_0} \leq t_{COMM}^{\Lcal_h}$, the $l$-th layer will be delayed for future assignment (Line 23 in Algorithm~\ref{algo:DreamDDP}).
\end{property}

\begin{property}\label{prop:AtLeastOne}
(\textbf{At-Least-One Assignment (ALO).}) If $\Lcal_h$ is $\emptyset$ when scheduling $l$-th layer at iteration $h$, assigning $A^l=h$ is optimal. The first case is that $t_{COMM}^l \leq t_{BP}^{\Lcal_{h:H}}-t_{BP}^{h_0}$, which satisfies property~\ref{prop:OptimalHiding}. In the other case $t_{COMM}^l > t_{BP}^{\Lcal_{h:H}}-t_{BP}^{h_0}$, assigning $A^l=h$ helps overlap $t_{COMM}^l$ with part of computation. Otherwise, delaying $A^l=h+1$ still causes the same overflowed communication time and narrowes the overlap opportunity of communicating subsequent layers.

In light of the ALO (property~\ref{prop:AtLeastOne}), when deciding $l$-th layer and $h$-th iteration, if $\Lcal_h$ is $\emptyset$, DreamDDP instantly assigns $A^l = h$ (Line 17 in Algorithm~\ref{algo:DreamDDP}).

In cases that all properties~\ref{prop:OptimalHiding}, ~\ref{prop:DelayedCO} and ~\ref{prop:AtLeastOne} cannot be utilized to reduce search space, the algorithm must conduct depth-first search (DFS) to collect the solutions 
\end{property}


\textbf{Complexity Analysis of Algorithm~\ref{algo:DreamDDP}.} We consider the worst case for the algorithm complexity, where the DFS search (Line 25 in Algorithm~\ref{algo:DreamDDP}) frequently happens during the scheduling. Each DFS search leads to two branches that need to be saved into the solution set. When $\Lcal_h$ is $\emptyset$, the $l$-th layer will be added into $\Lcal_h$ without searching, which does not generate more solutions. Thus, the size of $|\Omega|$ is less than $2^{L-H}$. Considering the branch $A^l=h$ (Line 28 in Algorithm~\ref{algo:DreamDDP}), the next layer $l-1$ must be delayed to assign to iteration $h+1$ because $ t_{BP}^{\Lcal_{h:H}}-t_{BP}^{h_0} <  t_{COMM}^{\Lcal_h\cup \{ l\}}$. Thus, the $|\Omega|$ is less than $2^H$. In summary, the size of the solution set $\Omega$ is upper bounded as $\Ocal(2^{\text{min}(L-H, H)})$.



\subsection{Filling the Bubble Time}\label{sec:bubble}

As illustated in Section~\ref{sec:partialLSGD} and Theorem~\ref{theo:converge-bound}, the model divergence $\Gamma_r = \frac{1}{K} \sum_{k=1}^K ||\wbar_{r} - w_{r}^{k} ||^2$ is a main influence factor of the convergence rate. Higher communication frequency (lower $H$) leads to less model divergence $\Gamma_r$ during training, thus accelerating convergence~\cite{kairouz2021advances,stich2018localsgd}. Orthogonal to increasing synchronization frequency $H$, which means to synchronize the whole model more frequently, another way to reduce the model divergence is to assign different layers with different communication frequencies. 

After scheduling the overlapping communication, we find that there may exist some idle communication time of the late layers during the late local training stages (larger local $r$ of in one synchronization period) as shown in Figure~\ref{fig:addbubble}. This implies that it is possible to synchronize these layers more frequently to reduce the model divergence, while not increasing the communication time.

\begin{figure}[!ht]
\vspace{-0.0cm}
\centering
\setlength{\abovecaptionskip}{0.2cm}
\setlength{\belowcaptionskip}{0.2cm} 
\includegraphics[width=0.95\linewidth]{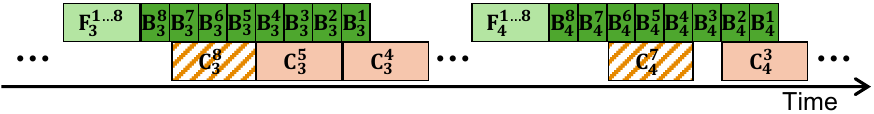}
\caption{The idle communication time of during the late local training stages that can be used to synchronize the late layers more frequently.}
\label{fig:addbubble}
\end{figure}

It is known that the early layers of deep learning models converge faster than the late layers~\cite{chen2023which,SVCCA} including the transformer~\cite{PipeTransformer}. This is because the early layers should be stablized during the late training stage, thus the late layers can be trained based on more stable intermediate features~\cite{chen2023which,SVCCA}. Fortunately, the idle communication time in the PLSGD aligns with the late layers as shown in Figure~\ref{fig:addbubble}. Thus, using the idle communication time to synchronize the late layers more frequently is exactly beneficial to reduce the model divergence than synchronizing more the early layers.


Formally, DreamDDP considers communicating more layers $\Lcal_h^{-} = \{L,...,l \}$ for communication iteration $h$ (Line 5 in Algorithm~\ref{algo:DreamDDP}) to accelerate convergence, while ensuring the communication time is not increased as following equation.

\begin{small}
\begin{equation}\label{eq:morelayers}
\text{max}(t_{BP}^{\Lcal_{h:H}}-t_{BP}^{h_0}, t_{COMM}^{\Lcal_h\cup\{L,...,l \}}) \leq \text{max}(t_{BP}^{\Lcal_{h:H}}-t_{BP}^{h_0}, t_{COMM}^{\Lcal_h}),
\end{equation}
\end{small}
in which $\{L,...,l\}$ disjoints with $\Lcal_h$.

\vspace{-0.0cm}
\section{Experimental Studies}\label{sec:exp}
\vspace{-0.1cm}
\subsection{Experimental Settings}\label{sec:exp-setting}
\textbf{Models and datasets.} We conduct distributed training on deep models including ResNet-18, ResNet-50, GPT-2 small~\cite{gpt2} and a small Llama-2~\cite{touvron2023llama} of 175M parameters. ResNet-18 is trained on CIFAR-10 and ResNet-50 is trained on CIFAR-100 with SGD momentum. Both GPT-2 and Llama-2 are trained on WikiText-2~\cite{wikitext} with Adam~\cite{KingBa15}. 

\textbf{Simulating Geo-distributed Training Clusters.} To simulate geo-distributed training, we vary the inter-machine bandwidths from 10MB/s to 1000GB/s to represent diverse network conditions across different geographic locations. ResNet-18 and ResNet-50 are trained on one cluster with 8 machines, each of which has 4 NVIDIA 2080Ti (totally 32 GPUs) and connected with 1GB/s Ethernet. LLMs are trained on the second cluster with 4 machines, each of which has 8 NVIDIA A6000 (totally 32 GPUs) and connected with 20GB/s Ethernet.

\textbf{Baselines.} We compare our algorithms with the naive S-SGD, ASC-WFBP~\cite{9488803} which is the state-of-the-art SGD overlapping algorithm as an improvement of WFBP based S-SGD~\cite{203269}, and LSGD with full synchronization (FLSGD)~\cite{coltLocalSGD,AperiodcLocalSGD}. As the ablation study, we also conduct experiments of the basic versions of DreamDDP, including LSGD with partial synchronization (PLSGD-ENP) of Equal-Number Partition (Example~\ref{ex:ENP}). Note that the optimizer for LLMs is Adam, which is generally used in the LLM training~\cite{gpt2,touvron2023llama}. The computation and communication dependency of stochastic Adam and SGD are same, as well as the local Adam and LSGD. Thus, we also call the stochastic Adam and local Adam as S-SGD and LSGD for simplicity. Here the S-SGD, LSGD and our algorithm mainly highlight the different modes of parallelization instead of the specific optimizers.



\subsection{Wall-clock Iteration Time}\label{sec:exp-wallclock}
Table~\ref{tab:wallclocktime} presents the wall-clock iteration time for each algorithm across different models and worker counts. 
DreamDDP consistently achieves the lowest iteration times across all models and worker configurations, demonstrating its superior efficiency in eliminating additional communication cost. Specifically, compared to the advanced DDP training method ASC-WFBP, DreamDDP achieves a speedup of $1.73 \sim 5.22\times$. Similarly, compared to FLSGD, DreamDDP achieves a speedup of $1.16 \sim 1.5\times$. This performance improvement can be even more pronounced in low-bandwidth environments, which are commonly encountered in geo-distributed scenarios.

\begin{table}[htb!]
\centering
\caption{Average iteration wall-clock time (in seconds) of 1000 running iterations. $S_1$ and $S_2$ represent the speedup of DreamDDP over ASC-WFBP and FLSGD.}
\vspace{-0.0cm}
\begin{adjustbox}{max width=\linewidth}
\begin{tabular}{c|cc|cc|cc|cc}
\hline
Model & \multicolumn{2}{c|}{ResNet-18} & \multicolumn{2}{c|}{ResNet-50} & \multicolumn{2}{c|}{GPT-2} & \multicolumn{2}{c}{Llama-2}\\
\# of workers & 8  & 32 & 8 & 32 & 8 & 32  & 8 & 32  \\
\hline
\hline
S-SGD/Adam & 1.31  & 2.40 & 3.08 & 4.52 & 2.63 & 8.67 & 2.80 &  8.75 \\
ASC-WFBP & 0.99 & 1.75 & 2.33 & 3.39 & 1.72 & 5.38& 1.83 & 5.36 \\
FLSGD/Adam & 0.66 & 0.57 & 1.59 & 0.94 & 0.80 & 2.08& 0.86 & 1.96   \\
PLSGD/Adam-ENP & 0.58 & 0.46 & 1.37 & 0.69 & 0.70 & 1.99& 0.77 & 1.89  \\
DreamDDP & \textbf{0.57} & \textbf{0.38} & \textbf{1.33} & \textbf{0.65} & \textbf{0.66} & \textbf{1.60} & \textbf{0.72} & \textbf{1.55}  \\
$S_1$ & $1.73\times$ & $4.61\times$ &$1.75\times$ & $5.22\times$ & $2.61\times$ & $3.36\times$ & $2.54\times$ & $3.46\times$  \\
$S_2$ & $1.16\times$ & $1.50\times$ &$1.20\times$ & $1.45\times$ & $1.21\times$ & $1.23\times$ & $1.20\times$ & $1.26\times$  \\
\hline
\end{tabular}
\end{adjustbox}
\label{tab:wallclocktime}
\vspace{-0.5cm}
\end{table}
\vspace{-0.0cm}

\begin{figure}[!ht]
\vspace{-0.0cm}
\centering
\setlength{\abovecaptionskip}{0.2cm}
\setlength{\belowcaptionskip}{0.2cm} 
    \subfigure[ResNet-18.]
    {
    \includegraphics[width=0.46\linewidth]{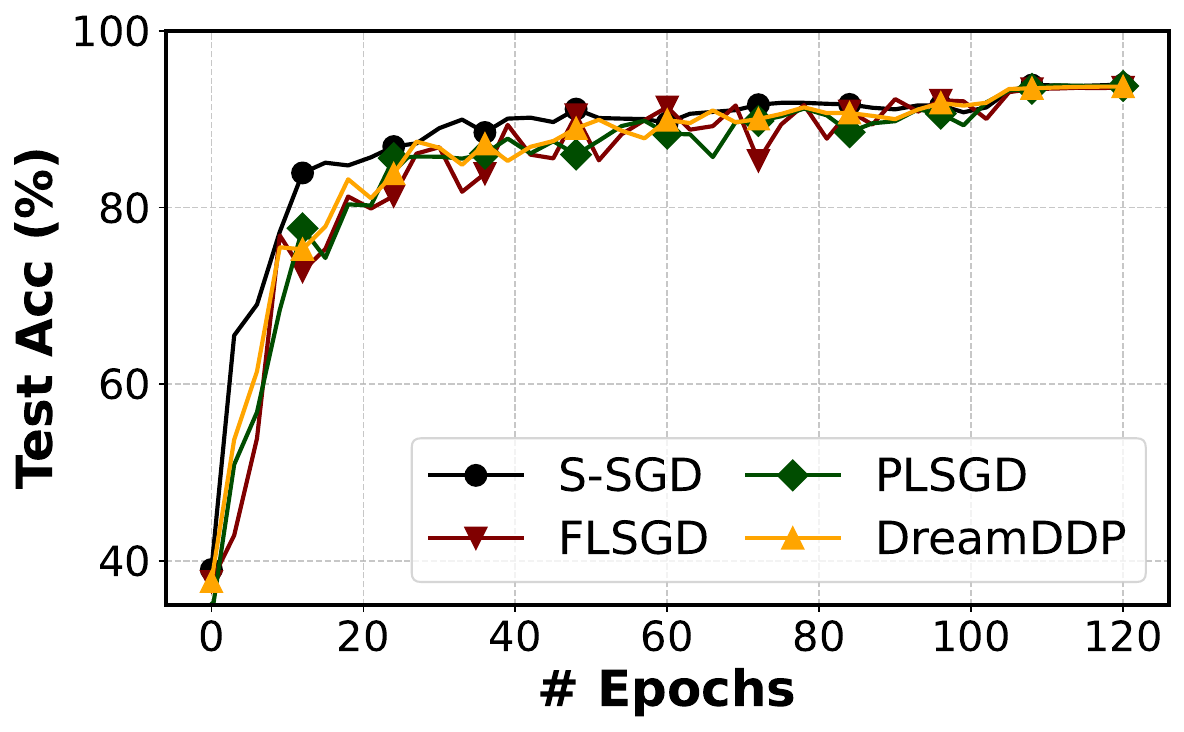}\label{fig:DifH-Res18-convergence-8w}
    }
    \subfigure[ResNet-50.]
    {
    \includegraphics[width=0.46\linewidth]{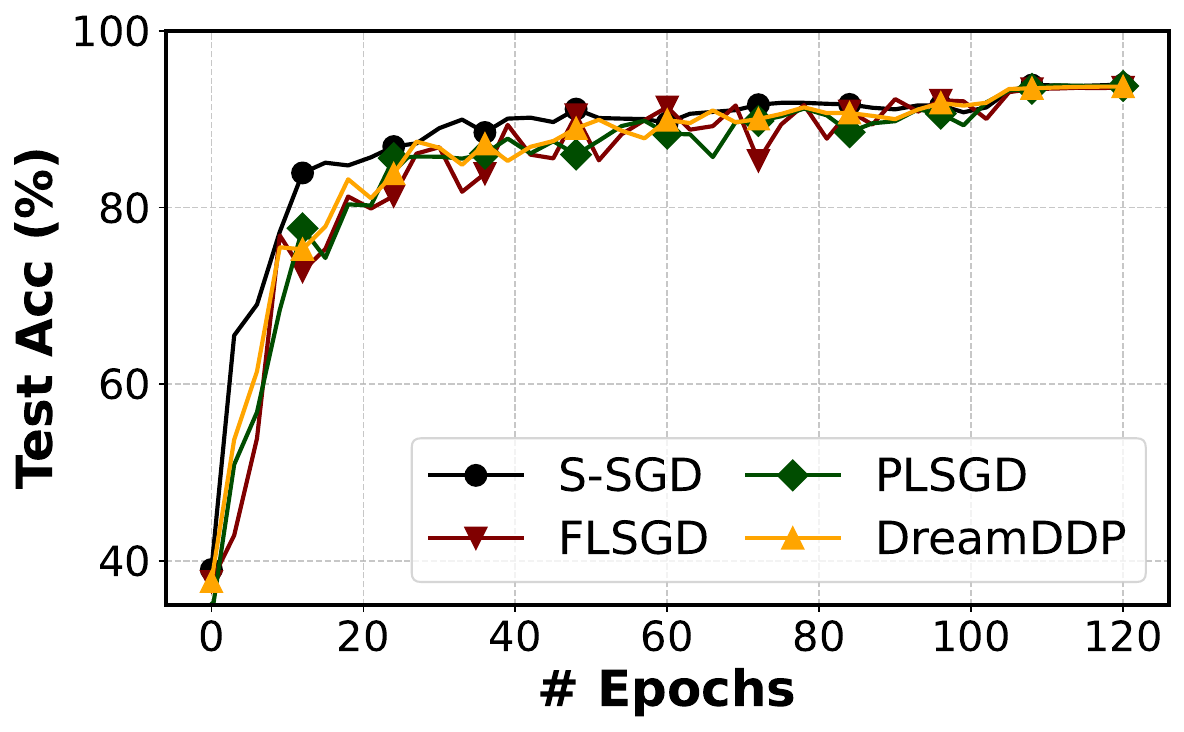}\label{fig:DifH-Res50-convergence-8w}
    }
    \subfigure[GPT-2.]
    {
    \includegraphics[width=0.46\linewidth]{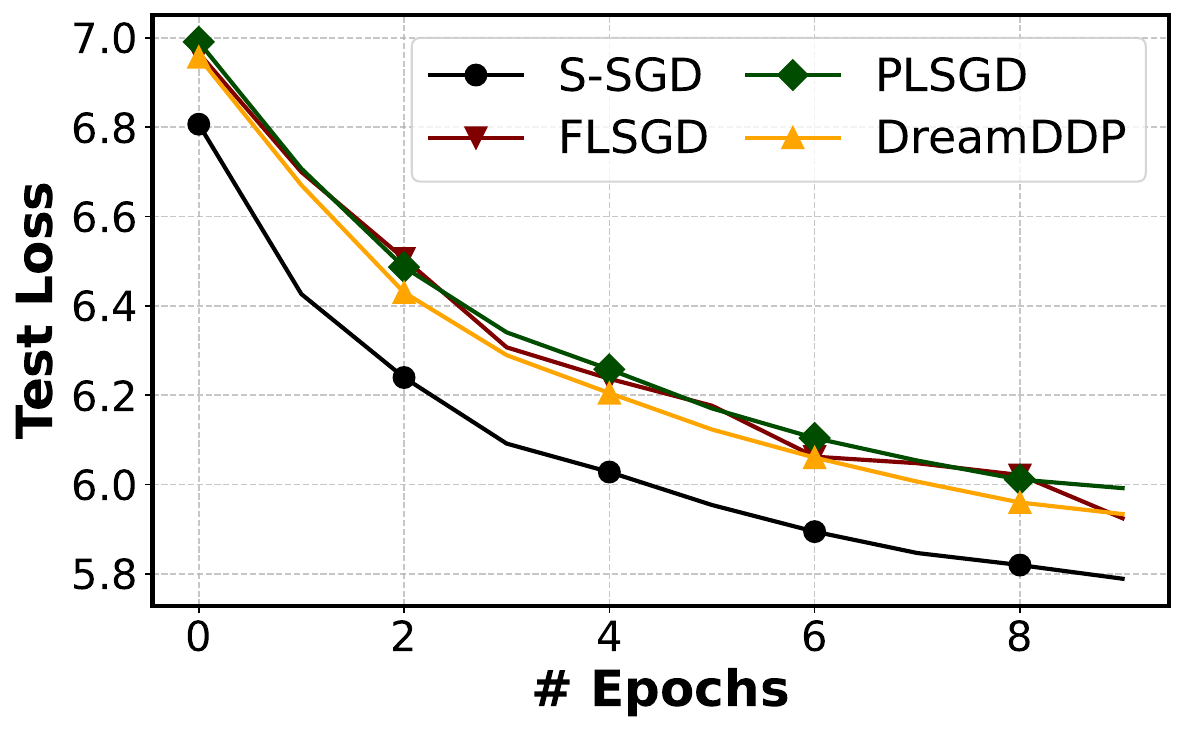}\label{fig:DifH-gpt2-convergence-8w}
    }
    \subfigure[Llama-2.]
    {
    \includegraphics[width=0.46\linewidth]{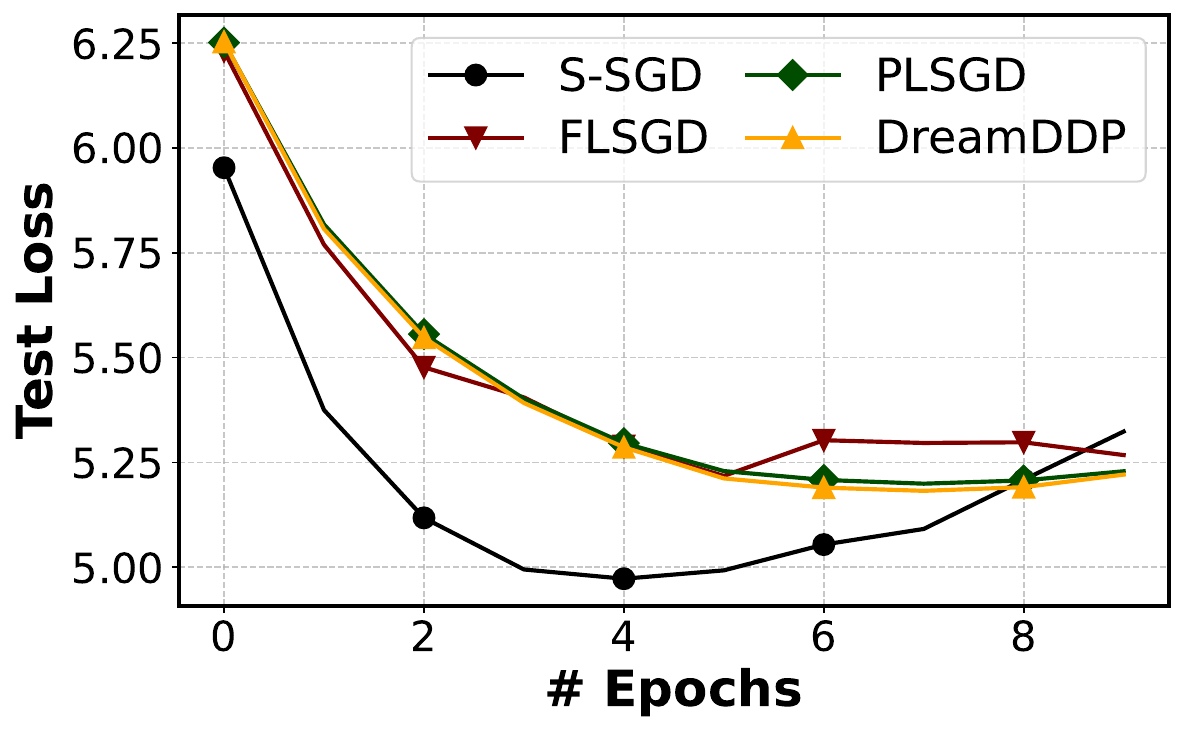}\label{fig:DifH-llama2-convergence-8w}
    }
    \caption{Convergence w.r.t. epochs on 8 workers.}
\label{fig:DifH-convergence-epochs-8w}
\vspace{-0.0cm}
\end{figure}

\begin{figure}[!ht]
\vspace{-0.0cm}
\centering
\setlength{\abovecaptionskip}{0.2cm}
\setlength{\belowcaptionskip}{0.2cm} 
    \subfigure[ResNet-18.]
    {
    \includegraphics[width=0.46\linewidth]{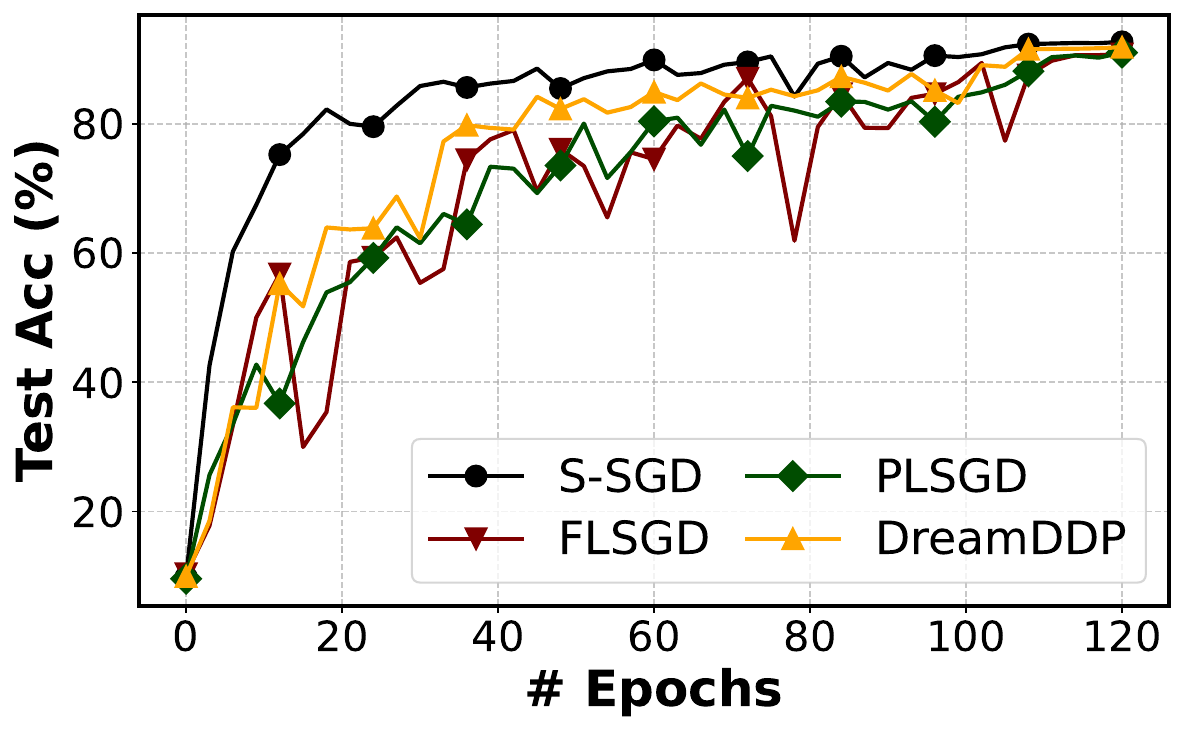}\label{fig:DifH-Res18-convergence}
    }
    \subfigure[ResNet-50.]
    {
    \includegraphics[width=0.46\linewidth]{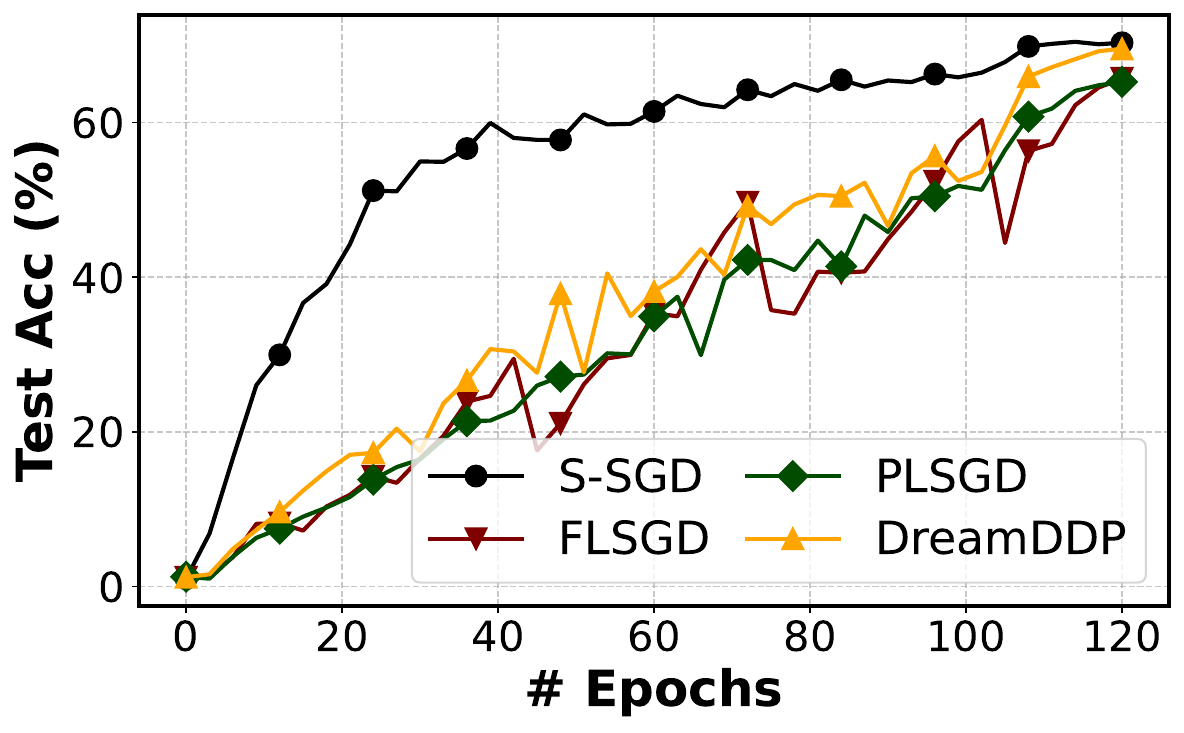}\label{fig:DifH-Res50-convergence}
    }
    \subfigure[GPT-2.]
    {
    \includegraphics[width=0.46\linewidth]{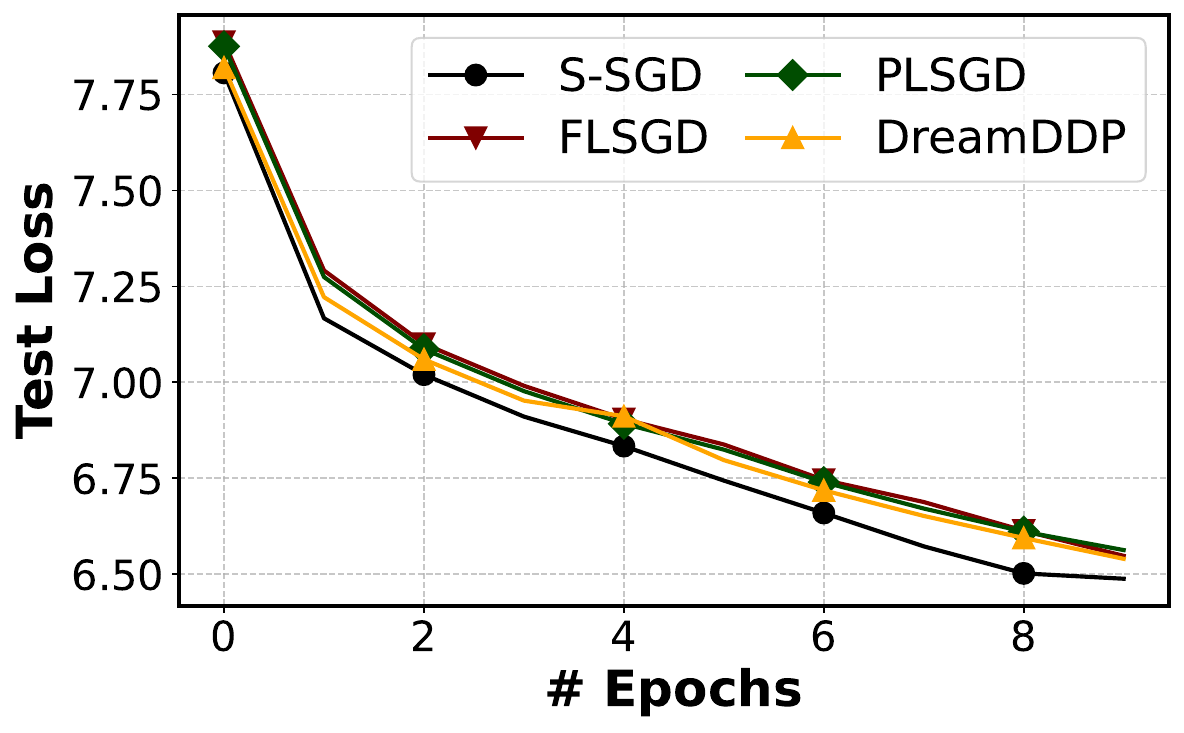}\label{fig:DifH-gpt2-convergence}
    }
    \subfigure[Llama-2.]
    {
    \includegraphics[width=0.46\linewidth]{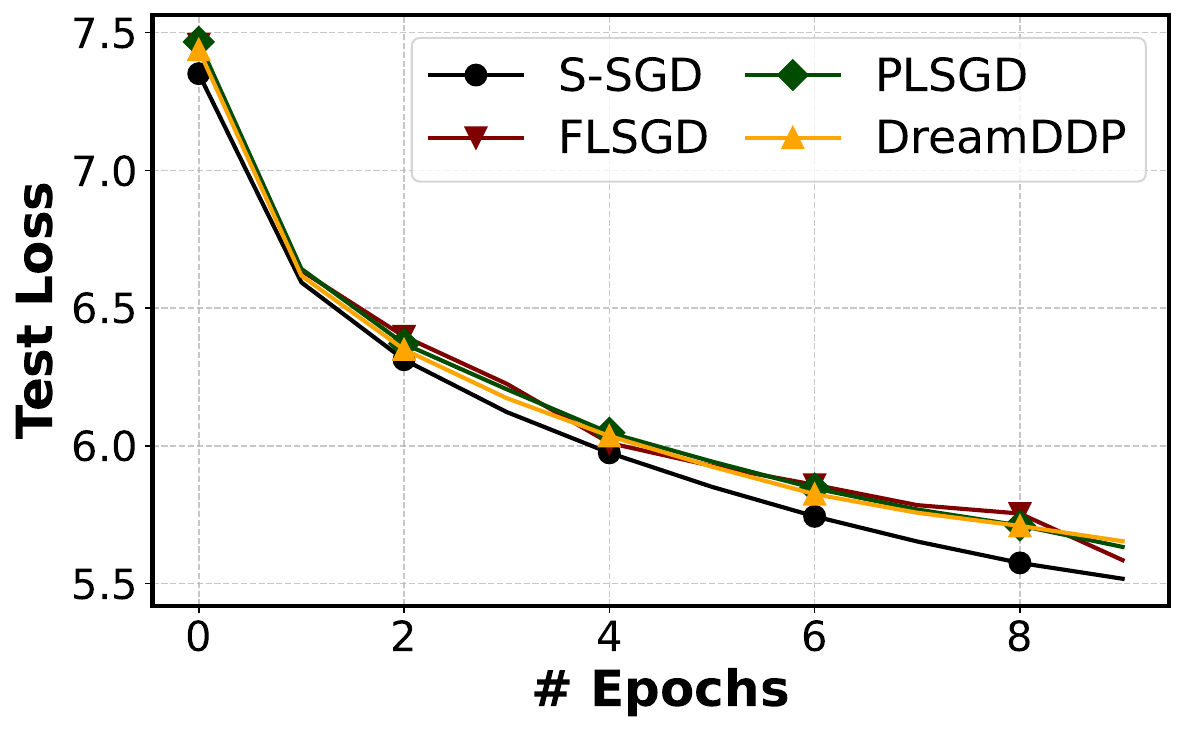}\label{fig:DifH-llama2-convergence}
    }
    \caption{Convergence w.r.t. epochs on 32 workers.}
\label{fig:DifH-convergence-epochs}
\vspace{-0.0cm}
\end{figure}

\subsection{Convergence w.r.t. Epoch}\label{sec:exp-convergence-epoch}
With respect to the iteration time, the FLSGD is also faster than ASC-WFBP. However, such acceleration introduces worse convergence due to its less $H\times$ communication size. We compare different synchronization frequency $H$ in ig.~\ref{fig:DifH-convergence-epochs-8w} and Fig.~\ref{fig:DifH-convergence-epochs}. Results show that the larger $H$ makes convergence w.r.t. epochs slower for both FLSGD and DreamDDP. Nevertheless, the supplement of the communication in DreamDDP helps accelerate the convergence. And the Fig.~\ref{fig:DifH-convergence-epochs} shows that the increased number of workers does not influence the performance improvement of DreamDDP, demonstrating the good scalability of it.



\begin{figure}[!ht]
\vspace{-0.0cm}
\centering
\setlength{\abovecaptionskip}{0.2cm}
\setlength{\belowcaptionskip}{0.2cm} 
    \subfigure[ResNet-18.]
    {
    \includegraphics[width=0.46\linewidth]{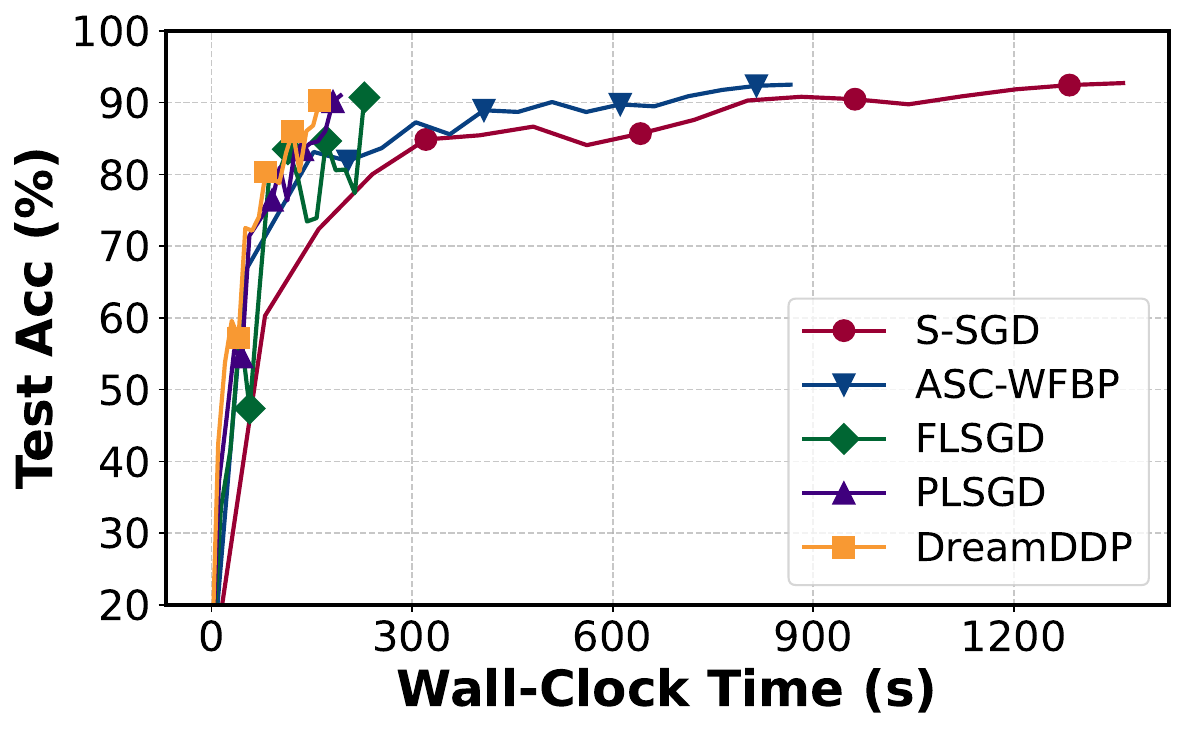}
    }
    \subfigure[ResNet-50.]
    {
    \includegraphics[width=0.46\linewidth]{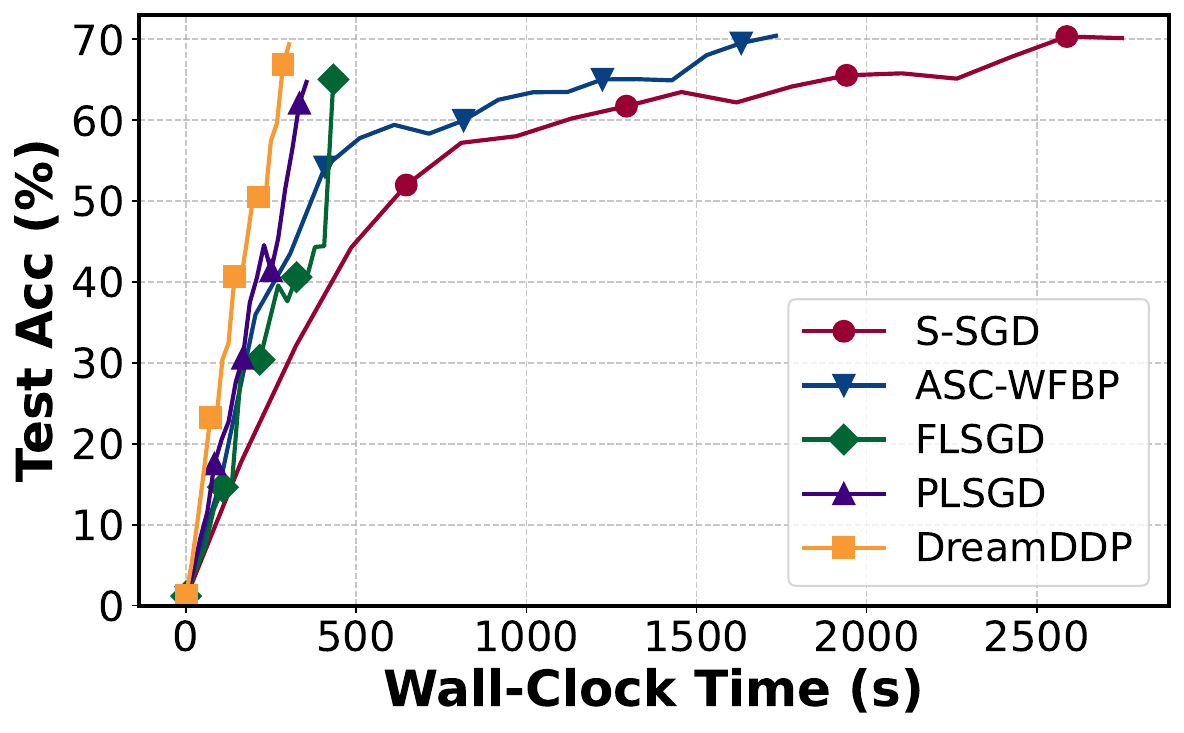}
    }
    \subfigure[GPT-2.]
    {
    \includegraphics[width=0.46\linewidth]{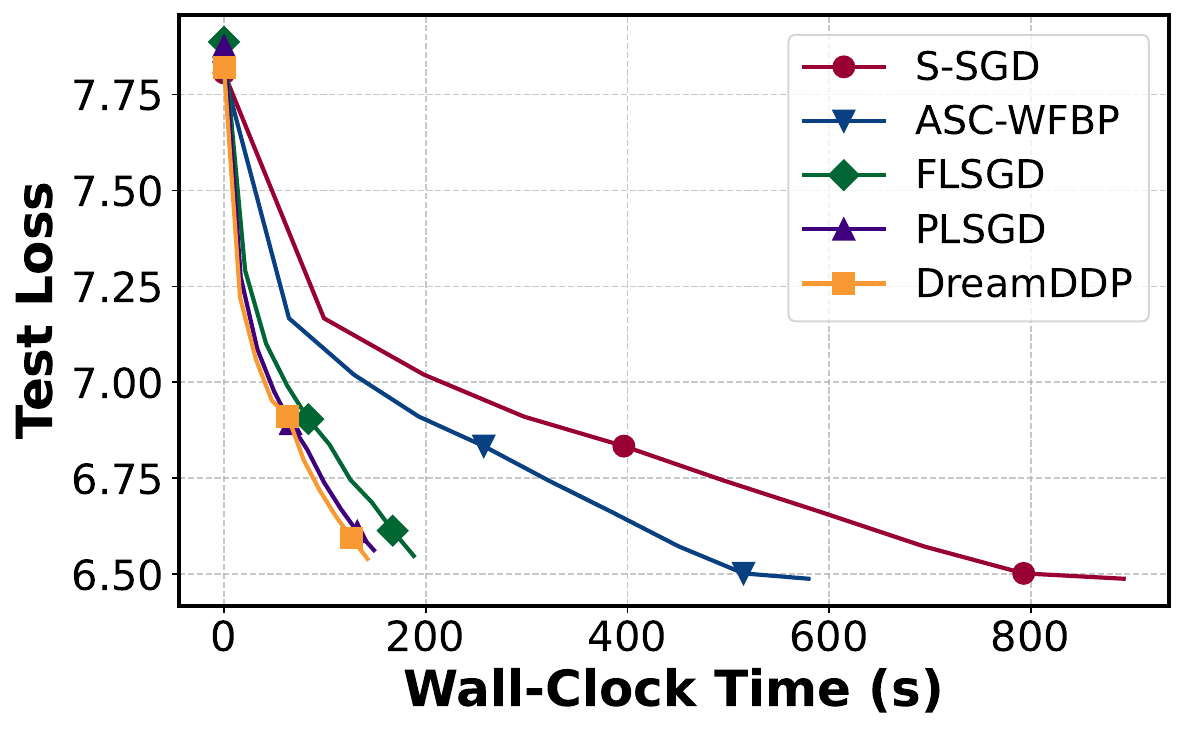}
    }
    \subfigure[Llama-2.]
    {
    \includegraphics[width=0.46\linewidth]{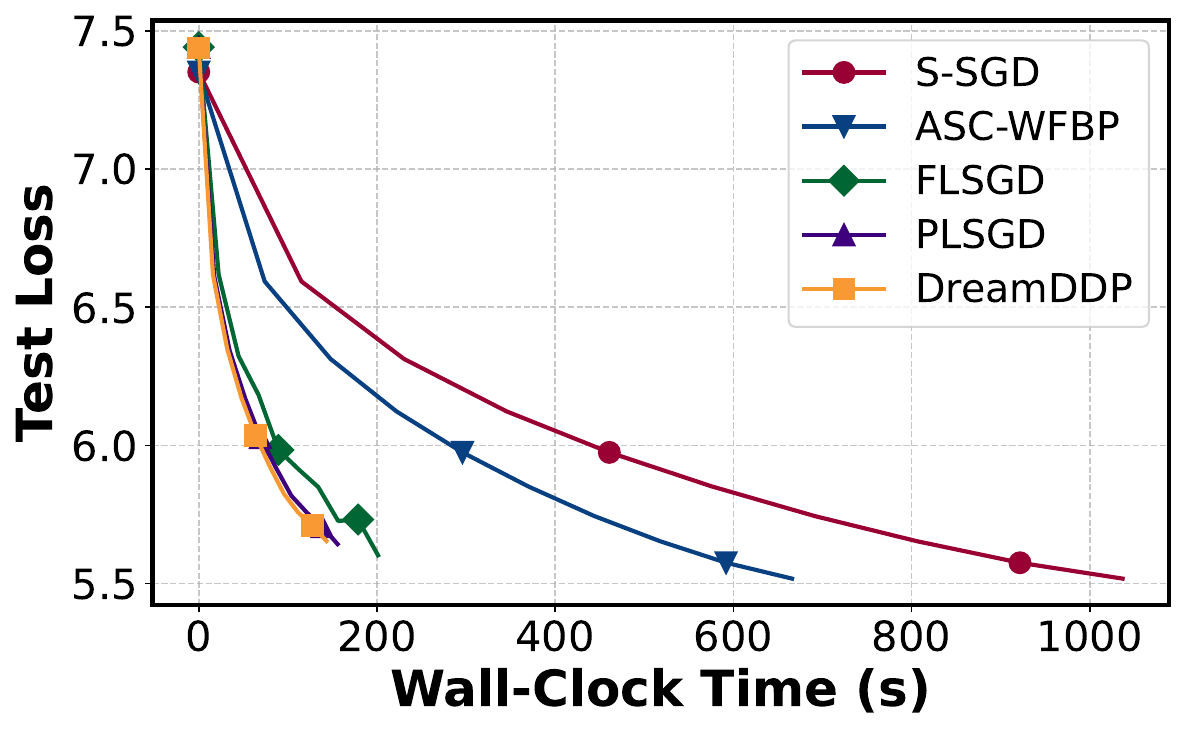}
    }
    \caption{Convergence w.r.t. wall-clock time of training with different algorithms on 32 workers.}
\label{fig:convergence-wallclock}
\vspace{-0.0cm}
\end{figure}

\subsection{Wall-clock Convergence}\label{sec:exp-convergence-time}
Fig.~\ref{fig:convergence-wallclock} illustrates the convergence of different algorithms w.r.t. wall-clock time and Table~\ref{tab:target-acc-wallclock} presents wall-clock time to reach target performance. Both visualizations show that the wall-clock convergence of all LSGD methods is significantly faster than that of SGD. The periodic model synchronization scheme effectively reduces the high communication cost associated with the frequent synchronization in S-SGD and ASC-WFBP. Specifically, DreamDDP achieves a maximum speedup of $3.91\times$ over ASC-WFBP.

Additionally, DreamDDP achieves a maximum speedup of $1.56\times$ over FLSGD, due to the advanced scheduling algorithm that eliminates the extra communication overhead. Furthermore, when compared to PLSGD with a similar communication-computation overlap strategy, DreamDDP demonstrates improved convergence speed, as displayed in Fig.~\ref{fig:convergence-wallclock} and Table~\ref{tab:target-acc-wallclock}. This improvement shows the effectiveness of DreamDDP's supplementary communication mechanism that maximally exploits overlapping opportunities without incurring extra training time.

\textbf{Different Synchronization Frequency.} We compare different DreamDDP synchronization frequencies $H$ with FLSGD with $H=10$ and S-SGD in Fig.~\ref{fig:convergence-epoches-8} and Fig.~\ref{fig:convergence-epoches-32}. The results on both 8 and 32 workers show that the larger $H$ makes convergence w.r.t. epochs slower for DreamDDP, which aligns with the intuition. Nevertheless, the supplement of the communication in DreamDDP helps accelerate the convergence, achieving comparable convergence speed with FLSGD with higher frequency. And the Fig.~\ref{fig:convergence-epoches-32} shows that the increased number of workers does not influence the performance improvement of DreamDDP, demonstrating the good scalability of it.

\begin{table}[htb!]
\centering
\caption{Wall-clock time to target model performance of different training methods. Different training methods require different wall-clock time to achieve the target performance.  For DL models, test accuracy is used. While for LLMs, training loss is used. $S_1$ and $S_2$ represent the speedup of DreamDDP over ASC-WFBP and FLSGD. We use the \textbf{bold font} to indicate the best results.}
\vspace{-0.0cm}
\begin{adjustbox}{max width=\linewidth}
\begin{tabular}{c|cc|cc|cc|cc}
\hline
Model & \multicolumn{2}{c|}{ResNet-18} & \multicolumn{2}{c|}{ResNet-50} & \multicolumn{2}{c|}{GPT-2} & \multicolumn{2}{c}{Llama-2}\\
\# of workers & 8  & 32 & 8 & 32 & 8 & 32  & 8 & 32  \\
\hline
\hline
\rowcolor{greyL} S-SGD/Adam & 1863.78  & 1282.08 & 6353.16  & 3546.24 & 394.15 & 167.29 & 342.57 &  247.55 \\
\rowcolor{greyL} ASC-WFBP & 1608.00 & 1420.26 & 5384.58 & 3192.24 & 361.19 & 148.32& 319.68 & 156.57   \\
\hline
FLSGD/Adam & 687.96 & 651.06 & 2141.64 & 1351.98 & 294.42 &114.38 & 213.15 &  127.10 \\
PLSGD/Adam-ENP & 603.36 & 653.04 & 1900.14 & 1261.08 & 270.68 &94.093 & 176.32 & 114.32 \\
DreamDDP & \textbf{573.10} & \textbf{515.58}& \textbf{1375.68} & \textbf{1067.76}  & \textbf{247.97} & \textbf{88.63} & \textbf{166.28} &  \textbf{105.07}\\
$S_1$ &  $2.80\times$ & $2.75\times$ & $3.91\times$ & $2.99\times$ &$1.46\times$ & $1.67\times$ & $1.92\times$ & $1.49\times$  \\
$S_2$ &$1.20\times$ &$1.26\times$ & $1.56\times$ & $1.27\times$ & $1.19\times$ & $1.29\times$ & $1.28\times$  &  $1.21\times$  \\
\hline
\end{tabular}
\end{adjustbox}
\label{tab:target-acc-wallclock}
\vspace{-0.5cm}
\end{table}

\begin{figure}[!ht]
\vspace{-0.0cm}
\centering
\setlength{\abovecaptionskip}{0.2cm}
\setlength{\belowcaptionskip}{0.2cm} 
    \subfigure[ResNet-18.]
    {
    \includegraphics[width=0.46\linewidth]{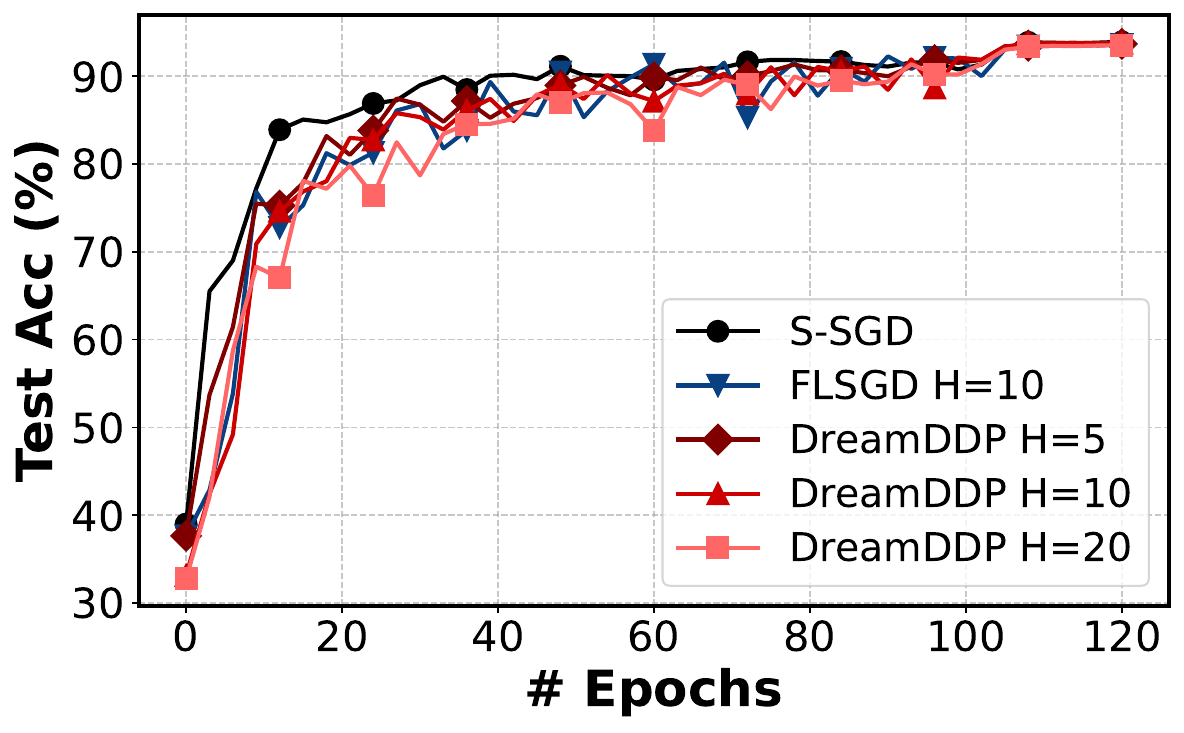}
    }
    \subfigure[ResNet-50.]
    {
    \includegraphics[width=0.46\linewidth]{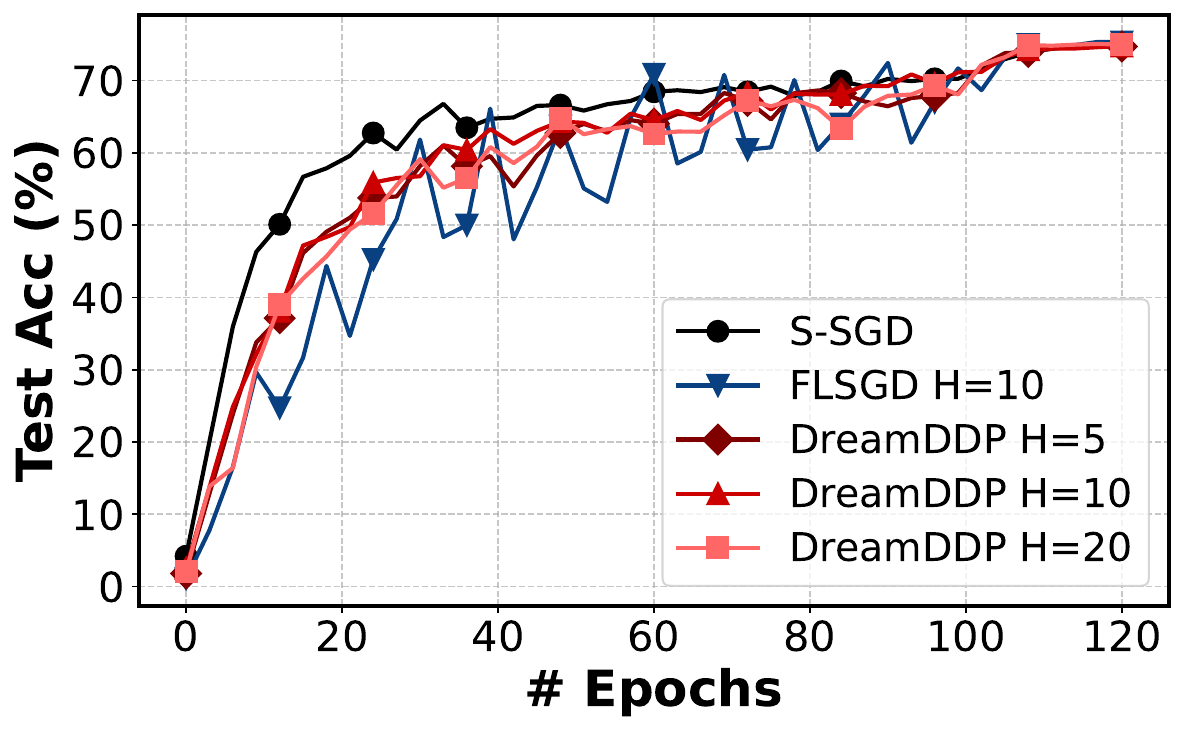}
    }
    \subfigure[GPT-2.]
    {
    \includegraphics[width=0.46\linewidth]{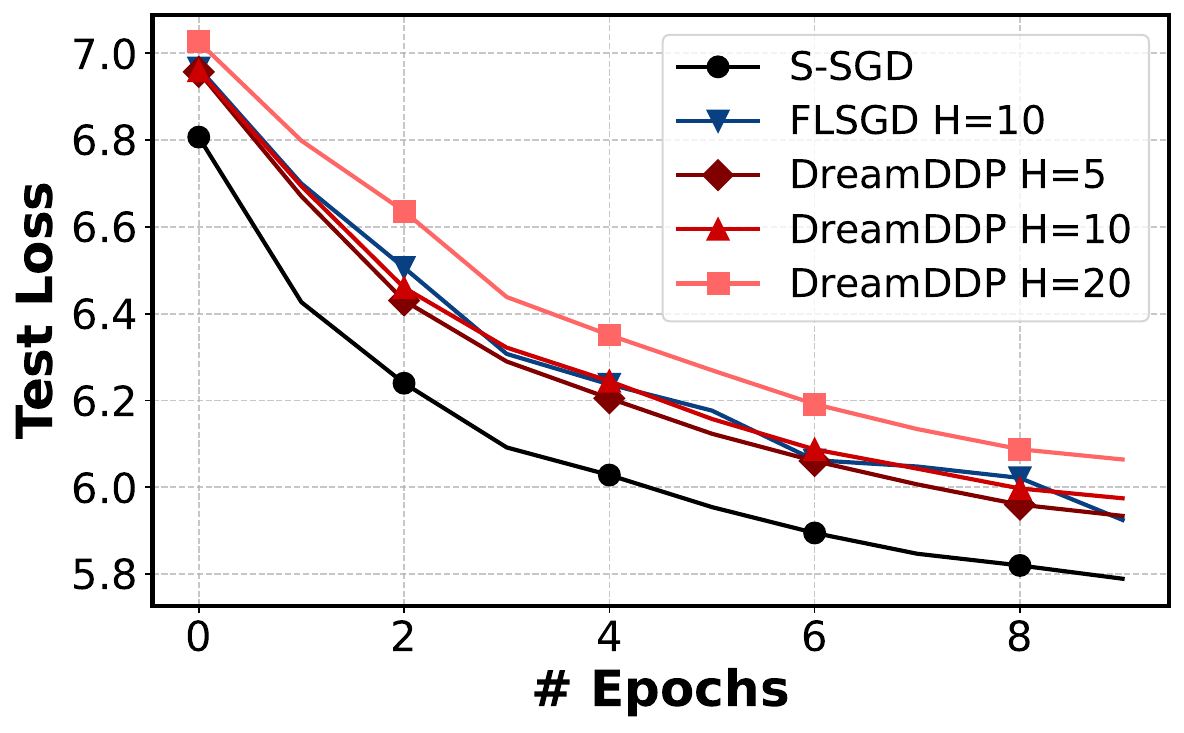}
    }
    \subfigure[Llama-2.]
    {
    \includegraphics[width=0.46\linewidth]{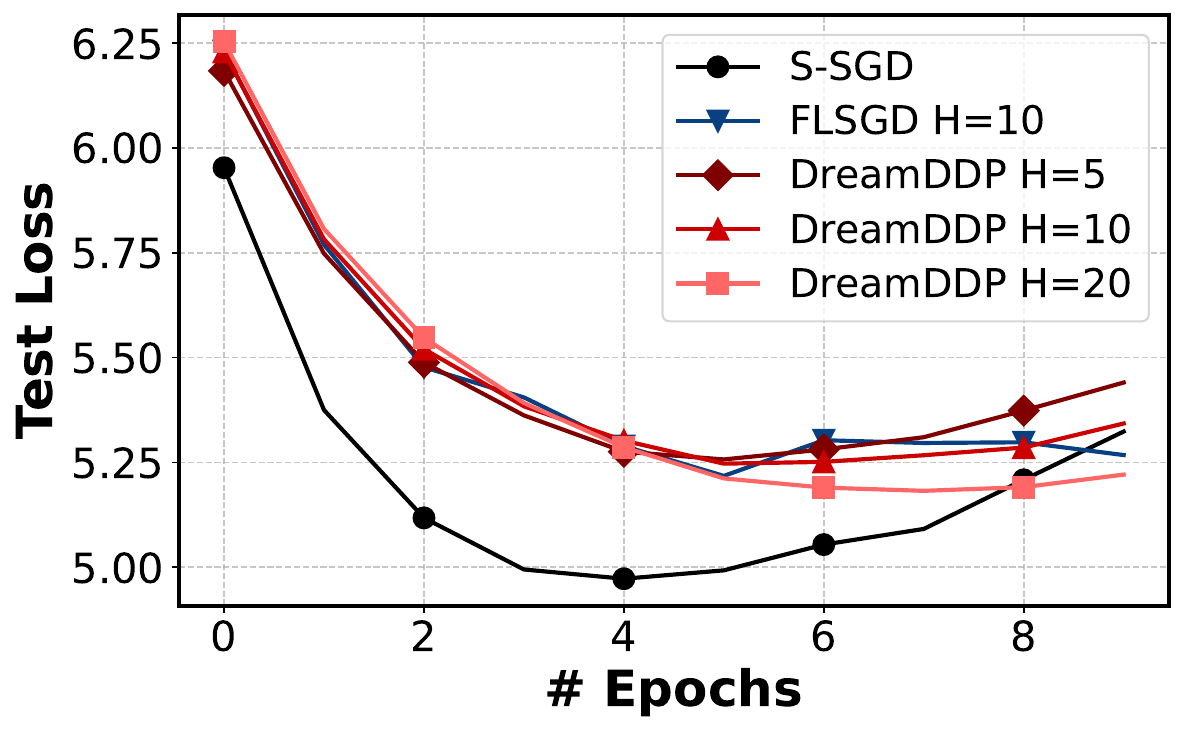}
    }
    \caption{Scheduling performance with different frequency $H$ on 8 workers.}
\label{fig:convergence-epoches-8}
\vspace{-0.0cm}
\end{figure}

\begin{figure}[!ht]
\vspace{-0.0cm}
\centering
\setlength{\abovecaptionskip}{0.2cm}
\setlength{\belowcaptionskip}{0.2cm} 
    \subfigure[ResNet-18.]
    {
    \includegraphics[width=0.46\linewidth]{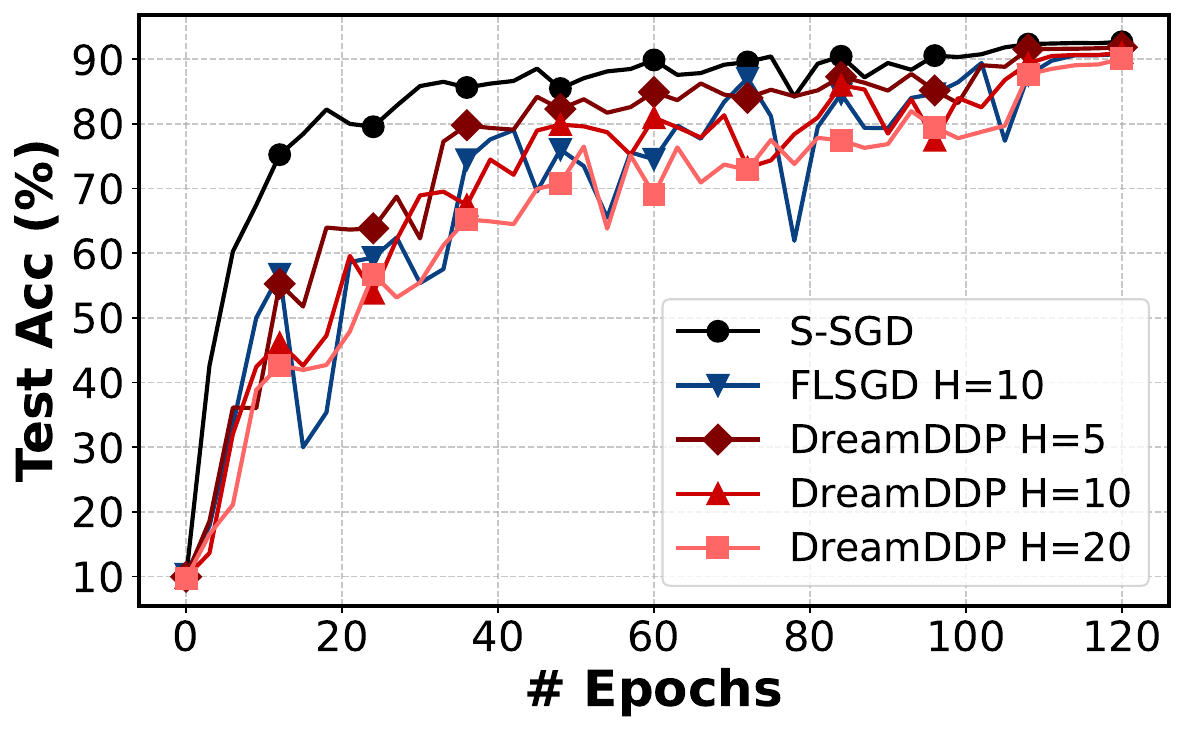}
    }
    \subfigure[ResNet-50.]
    {
    \includegraphics[width=0.46\linewidth]{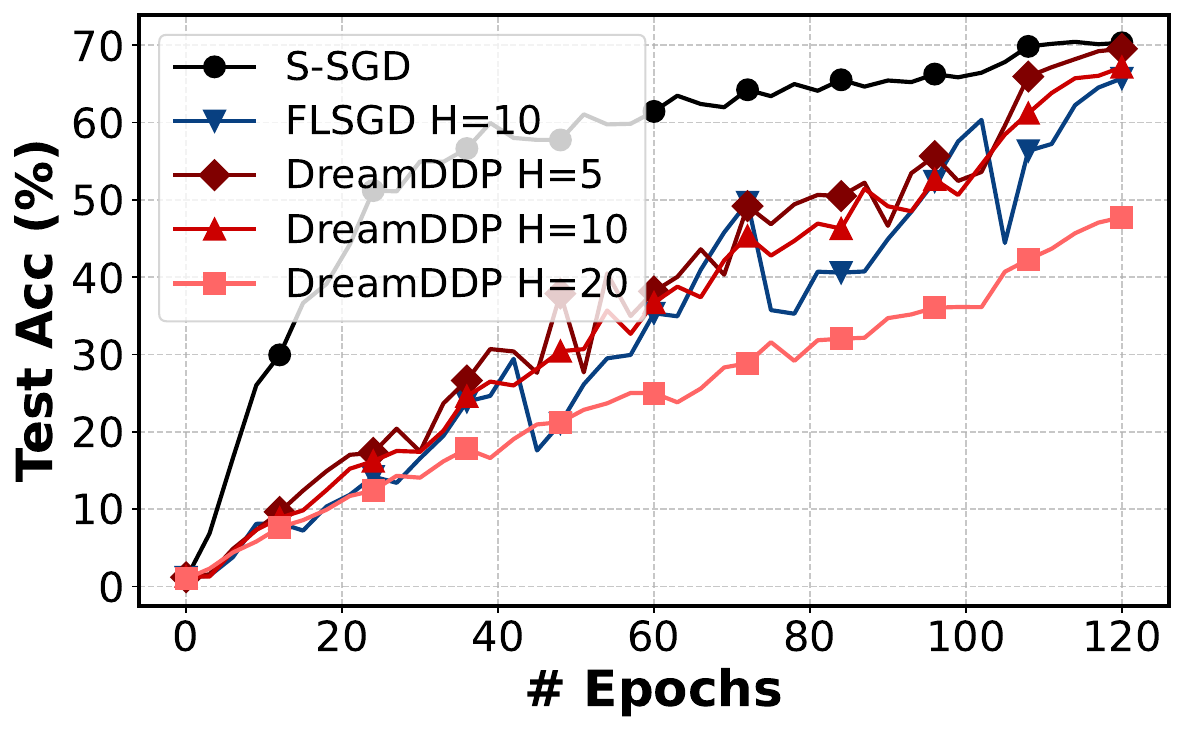}
    }
    \subfigure[GPT-2.]
    {
    \includegraphics[width=0.46\linewidth]{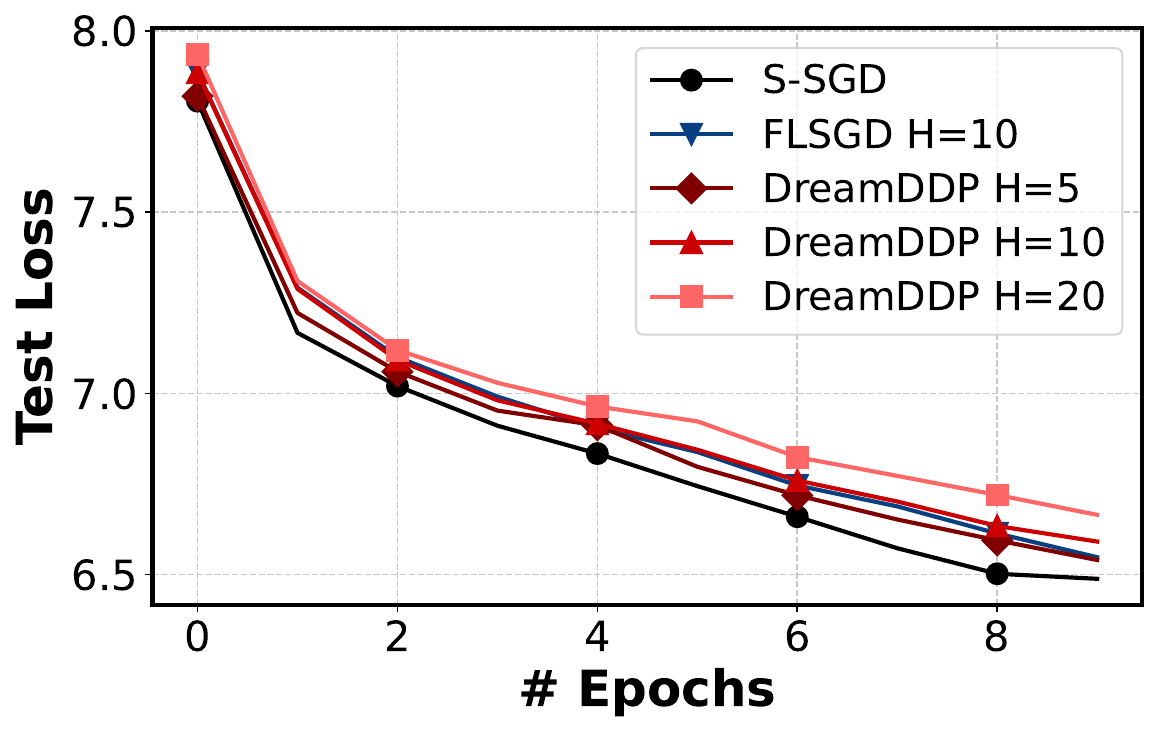}
    }
    \subfigure[Llama-2.]
    {
    \includegraphics[width=0.46\linewidth]{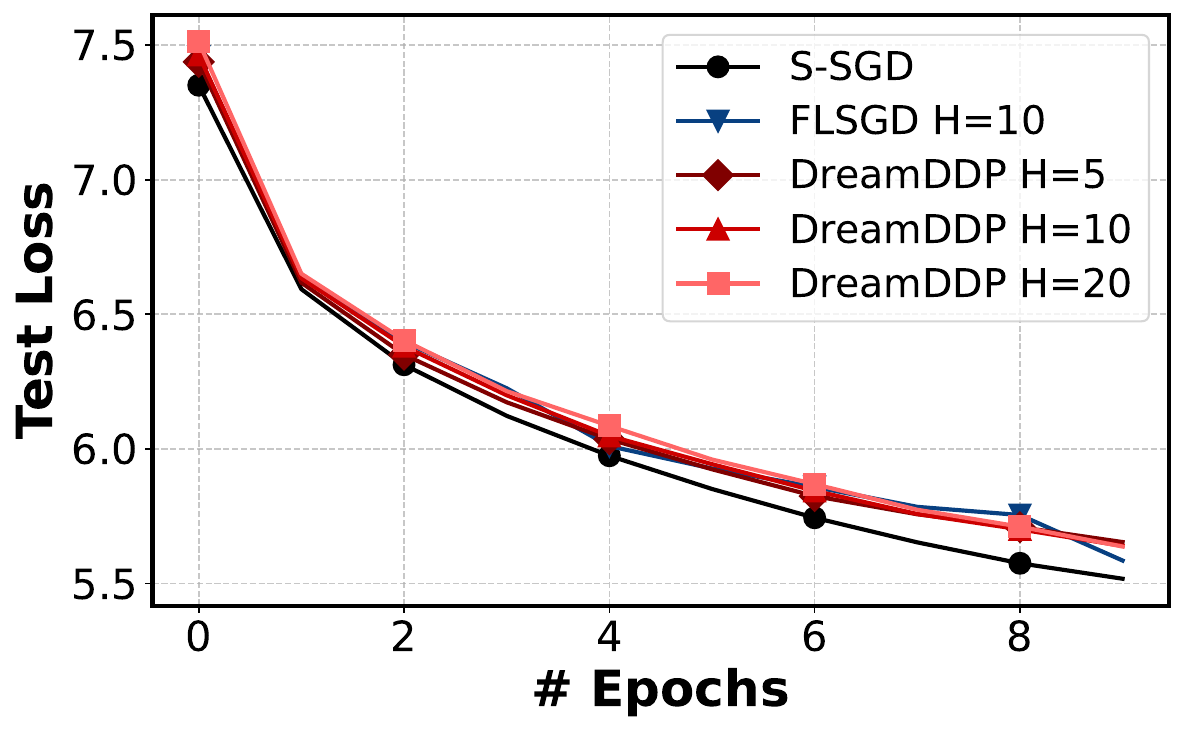}
    }
    \vspace{-0.0cm}
    \caption{Scheduling performance with different frequency $H$ on 32 workers.}
\label{fig:convergence-epoches-32}
\vspace{-0.0cm}
\end{figure}

\subsection{Scheduling Performance}\label{sec:exp-scheduling performance}

\textbf{Different number of layers.} To explore whether DreamDDP can schedule to overlap the communication time to a maximum degree, we compare the non-overlapped time of it with the brute force search. As the time complexity of the brute force is too large ($\frac{(H+L)!}{L!H!}$), we show the results of scheduling with up to 30 layers with $H=5$. Fig.~\ref{fig:extra_time} presents a comparison of the extra communication time not overlapped with computation after scheduling with DreamDDP and the brute force. Fig.~\ref{fig:waittime-layers} shows that with varying number of layers, during a single training cycle (5 iterations), DreamDDP performs equivalently to the brute force optimal solution in most cases, with only slightly longer extra communication times in a few scenarios. 

\textbf{Different bandwidth.} Fig.~\ref{fig:waittime-bandwidth} shows varying bandwidth, DreamDDP also has different performance with the optimal results searched by the brute force. These results demonstrate that DreamDDP effectively identifies near-optimal scheduling strategies while significantly reducing search complexity.

\begin{figure}[!ht]
\centering
\setlength{\abovecaptionskip}{0.cm}
\setlength{\belowcaptionskip}{-0.2cm} 
    \subfigure[different \# of layers]
    {
    \includegraphics[width=0.46\linewidth]{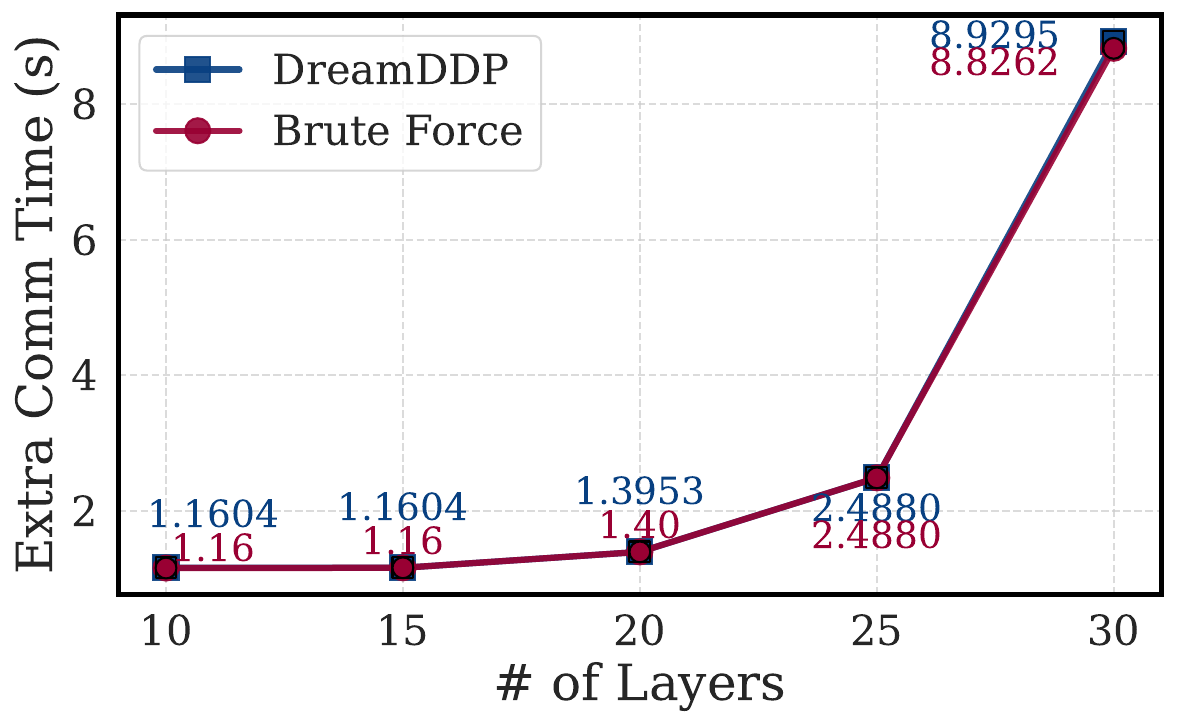}\label{fig:waittime-layers}
    }
    \subfigure[different bandwidths]
    {
    \includegraphics[width=0.46\linewidth]{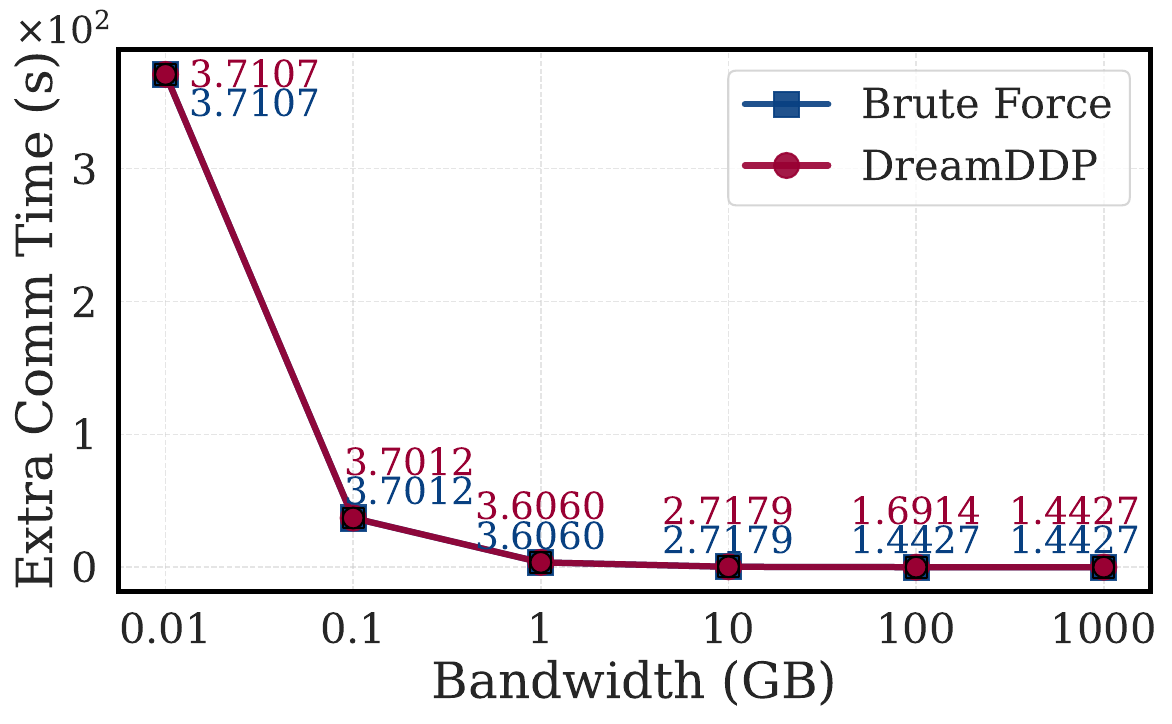}\label{fig:waittime-bandwidth}
    }
    \caption{Comparison of extra communication time not overlapped with computation between.}
\label{fig:extra_time}
\vspace{-0.1cm}
\end{figure}

\subsection{Profiling and Searching Overhead}\label{sec:exp-profiling}
We compare the theoretical and realistic profiled running time of the scheduling algorithms, including DreamDDP and brute force searching on scheduling different models. Figure~\ref{fig:theory-sche-time} shows the theoretical time complexity of scheduling four models with $H=5$. As these models have too many layers, the time complexity of scheduling them with the brute force is infeasible. Thus, we only choose 30 layers for these four models to compare the realistic running time of them. Figure~\ref{fig:realistic-sche-time} shows that DreamDDP's scheduling method significantly reduces the searching complexity by several orders of magnitude, compared with brute force method. DreamDDP reduces the search space, effectively eliminating the need to explore all possible layer partitions.

\begin{figure}[!ht]
\vspace{-0.4cm}
\centering
\setlength{\abovecaptionskip}{0.cm}
\setlength{\belowcaptionskip}{-0.2cm} 
    \subfigure[Theoretical Time Complexity]
    {
    \includegraphics[width=0.46\linewidth]{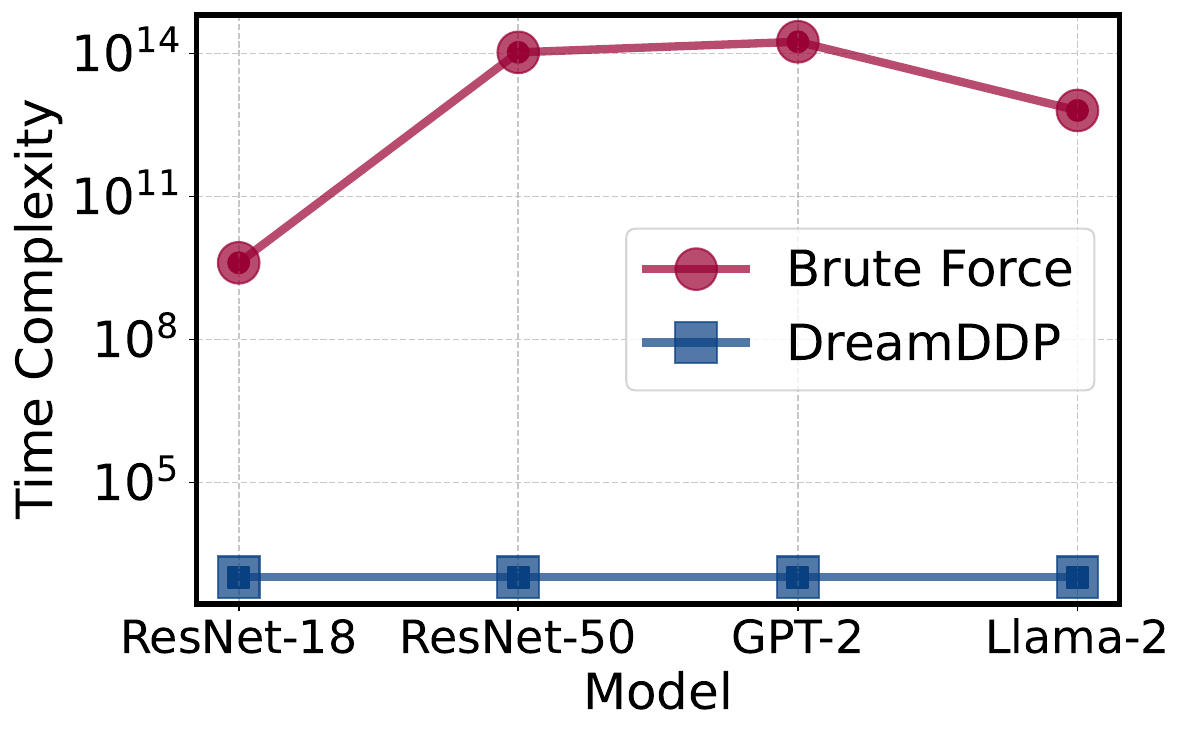}\label{fig:theory-sche-time}
    }
    \subfigure[Real Scheduling Time]
    {
    \includegraphics[width=0.46\linewidth]{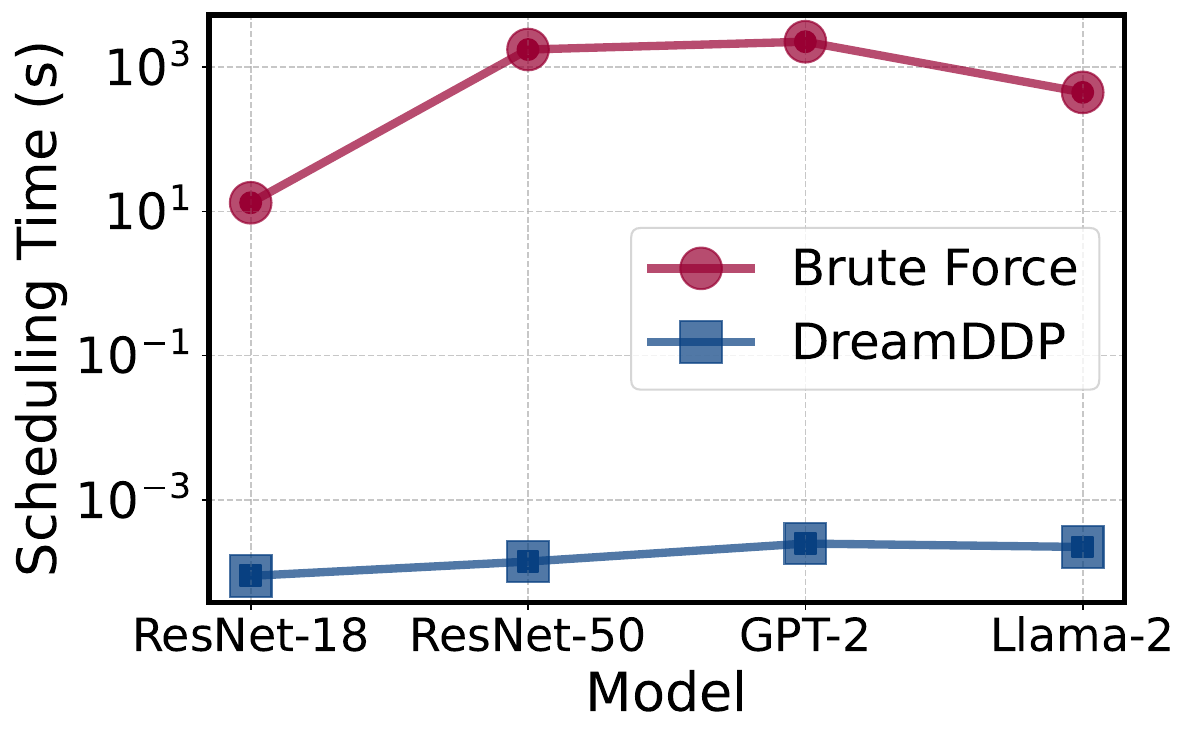}\label{fig:realistic-sche-time}
    }
    \caption{Time complexity and running time of scheduling algorithms.}
\label{fig:Timecomplexity}
\vspace{-0.1cm}
\end{figure}


\section{Related Works}\label{sec:related}


\textbf{Overlapping Communication and Computation.}
In distributed training with S-SGD, the gradient of each layer can be communicated instantly after its BP and overlapped with BP of subsequent layers, thus saving total training time~\cite{203269}. Overlapping communication with feed-forward computing~\cite{10.1145/3341301.3359642} helps further reduce total training time. Some observe that merging small tensors for communication~\cite{9155269} can help improve communication efficiency. ASC-WFBP~\cite{9488803} further improves overlapping by simultaneous communications and scheduling. Different from them, we consider overlapping communication and computation in LSGD.


\textbf{LSGD.}
LSGD~\cite{stich2018localsgd,woodworth2020minibatch,AperiodcLocalSGD,sadiev2022communication} reduces the communication time by periodically averaging model weights instead of gradients. The synchronization frequency strongly influences the convergence and communication time. Thus, some works vary synchronization periods to tradeoff convergence and communication~\cite{stich2018localsgd,AperiodcLocalSGD,wang2019adaptive}. The strong synchronization severely restricts the potential improvements of LSGD, preventing it from communication overlapping. To this end, we propose partial synchronization to decouple the whole model and synchronize layers at different training iterations.


The data distribution in different workers also influences the convergence rate of LSGD~\cite{woodworth2020minibatch}. In this work, we consider the IID data in centralized distributed training rather than federated Learning. Nevertheless, our method and the convergence analysis can be easily extended to the federated learning scenarios~\cite{kairouz2021advances,VHL}.

\textbf{Reducing Communication Costs.}
To reduce the communication time, some works propose compressing the communicated gradients~\cite{10.1145/3452296.3472904,tang2020survey} in S-SGD or model parameters~\cite{10229032,10229086,GossipFL} in LSGD. However, the compression may require large potential computation time and slow convergence, which prevents it from deployment~\cite{agarwal2022utility}. ~\cite{10229032,10229086} proposed compressing local results of Federated Learning~\cite{VHL,GossipFL}, in which LSGD is used. Orthogonal to this direction, our method aims to modify LSGD to improve system efficiency.




\vspace{-0.1cm}
\section{Limitations}\label{sec:limitations}
\vspace{-0.1cm}



\textbf{Heterogeneous Environments.} The real-world geo-distributed settings may involve heterogeneous hardware configurations and bandwidth, leading to various computing and communication time between different clusters. Thus, we may need to consider how to design suitable algorithms and scheduling to address this problem.

\textbf{Dynamic Layer Partitioning.} Our method relies on profiled computation and communication times to optimize scheduling. However, in geo-distributed training scenarios where bandwidth can change frequently, the static profiled times may no longer accurately reflect the current network conditions, potentially reducing the effectiveness of the scheduling strategy.


\vspace{-0.1cm}
\section{Conclusion}\label{sec:conclusion}
\vspace{-0.1cm}
In geo-distributed training, traditional S-SGD suffers from the communication costs of aggregating gradients, especially in low-bandwidth environments and LLMs. LSGD reduces the communication cost, but it requires strict synchronization and suffers from slow convergence speed. In this work, we first propose partial synchronization to relax the strict synchronization of LSGD. Then we design DreamDDP to schedule the communication overlapped with backpropagation computation to reduce the communication overhead of LSGD. DreamDDP also supplements communication with maximally exploiting overlapping without extra training time. The experimental results show that DreamDDP accelerates both wall-clock iteration time and the convergence speed than ASC-WFBP and LSGD (and Adam) on four popular types of DL models including ResNets, GPT, and Llama-2.

\bibliography{cite}

\begin{thebibliography}{10}

\bibitem{agarwal2022utility}
S.~Agarwal, H.~Wang, S.~Venkataraman, and D.~Papailiopoulos.
\newblock On the utility of gradient compression in distributed training systems.
\newblock {\em Proceedings of Machine Learning and Systems}, 4:652--672, 2022.

\bibitem{268961}
J.~Albrecht and C.~Tuttle.
\newblock Loose synchronization for {Large-Scale} networked systems.
\newblock In {\em 2006 USENIX Annual Technical Conference (USENIX ATC 06)}, 2006.

\bibitem{ansel2024pytorch}
J.~e.~a. Ansel.
\newblock Pytorch 2: Faster machine learning through dynamic python bytecode transformation and graph compilation.
\newblock In {\em ASPLOS}, pages 929--947, 2024.

\bibitem{deepspeed1}
M.~Artetxe, S.~Bhosale, N.~Goyal, T.~Mihaylov, M.~Ott, S.~Shleifer, X.~V. Lin, J.~Du, S.~Iyer, R.~Pasunuru, G.~Anantharaman, X.~Li, S.~Chen, H.~Akin, M.~Baines, L.~Martin, X.~Zhou, P.~S. Koura, B.~O'Horo, J.~Wang, L.~Zettlemoyer, M.~T. Diab, Z.~Kozareva, and V.~Stoyanov.
\newblock Efficient large scale language modeling with mixtures of experts.
\newblock {\em CoRR}, abs/2112.10684, 2021.

\bibitem{Bottou2016OptimizationMF}
L.~Bottou, F.~E. Curtis, and J.~Nocedal.
\newblock Optimization methods for large-scale machine learning.
\newblock {\em SIAM Review}, 60:223--311, 2016.

\bibitem{chen2023which}
Y.~Chen, A.~Yuille, and Z.~Zhou.
\newblock Which layer is learning faster? a systematic exploration of layer-wise convergence rate for deep neural networks.
\newblock In {\em ICLR}, 2023.

\bibitem{cheng2024editlocalsgdbasedefficientdistributed}
J.~Cheng, N.~Gao, Y.~Yue, Z.~Ye, J.~Jiang, and J.~Sha.
\newblock Edit: A local-sgd-based efficient distributed training method for large language models.
\newblock {\em arXiv preprint arXiv:2412.07210}, 2024.

\bibitem{ClearMLConsumerGPUs}
Clear.ml.
\newblock Consumer gpus vs datacenter gpus for cv.
\newblock [Online].
\newblock https://clear.ml/blog/consumer-gpus-vs-datacenter-gpus-for-cv-the-surprising-cost-effective-winner.

\bibitem{DeDLOC}
M.~Diskin, A.~Bukhtiyarov, M.~Ryabinin, L.~Saulnier, A.~Sinitsin, D.~Popov, D.~V. Pyrkin, M.~Kashirin, A.~Borzunov, A.~Villanova~del Moral, et~al.
\newblock Distributed deep learning in open collaborations.
\newblock In {\em NeurIPS}, 2021.

\bibitem{douillard2024diloco}
A.~Douillard, Q.~Feng, A.~A. Rusu, R.~Chhaparia, Y.~Donchev, A.~Kuncoro, M.~Ranzato, A.~Szlam, and J.~Shen.
\newblock Diloco: Distributed low-communication training of language models, 2024.

\bibitem{kairouz2021advances}
P.~K. et~al.
\newblock Advances and open problems in federated learning, 2021.

\bibitem{10.1145/3452296.3472904}
J.~Fei, C.-Y. Ho, A.~N. Sahu, M.~Canini, and A.~Sapio.
\newblock Efficient sparse collective communication and its application to accelerate distributed deep learning.
\newblock In {\em ACM SIGCOMM}, SIGCOMM '21, page 676–691, 2021.

\bibitem{10229032}
Y.~Guan, X.~Liu, T.~Ren, and J.~Niu.
\newblock Enabling communication-efficient federated learning via distributed compressed sensing.
\newblock In {\em INFOCOM}, pages 1--10, 2023.

\bibitem{guo2022hybrid}
Y.~Guo, Y.~Sun, R.~Hu, and Y.~Gong.
\newblock Hybrid local sgd for federated learning with heterogeneous communications.
\newblock In {\em ICLR}, 2022.

\bibitem{PipeTransformer}
C.~He, S.~Li, M.~Soltanolkotabi, and S.~Avestimehr.
\newblock Pipetransformer: Automated elastic pipelining for distributed training of large-scale models.
\newblock In {\em ICML}, 2021.

\bibitem{he2024exploring}
Q.~He, X.~Zhuang, and Z.~Wu.
\newblock Exploring scaling laws for local sgd in large language model training.
\newblock {\em arXiv preprint arXiv:2409.13198}, 2024.

\bibitem{jaghouar2024intellect1technicalreport}
S.~Jaghouar, J.~M. Ong, M.~Basra, F.~Obeid, J.~Straube, M.~Keiblinger, E.~Bakouch, L.~Atkins, M.~Panahi, C.~Goddard, M.~Ryabinin, and J.~Hagemann.
\newblock Intellect-1 technical report, 2024.

\bibitem{OpenDiLoCo}
S.~Jaghouar, J.~M. Ong, and J.~Hagemann.
\newblock Opendiloco: An open-source framework for globally distributed low-communication training, 2024.

\bibitem{jiang2024mixtral}
A.~Q. Jiang, A.~Sablayrolles, A.~Roux, A.~Mensch, B.~Savary, C.~Bamford, D.~S. Chaplot, D.~d.~l. Casas, E.~B. Hanna, F.~Bressand, et~al.
\newblock Mixtral of experts.
\newblock {\em arXiv preprint arXiv:2401.04088}, 2024.

\bibitem{jiang2020unified}
Y.~Jiang, Y.~Zhu, C.~Lan, B.~Yi, Y.~Cui, and C.~Guo.
\newblock A unified architecture for accelerating distributed $\{$DNN$\}$ training in heterogeneous $\{$GPU/CPU$\}$ clusters.
\newblock In {\em OSDI}, pages 463--479, 2020.

\bibitem{kairouz2019advances}
P.~Kairouz, H.~B. McMahan, B.~Avent, A.~Bellet, M.~Bennis, A.~N. Bhagoji, K.~Bonawitz, Z.~Charles, G.~Cormode, R.~Cummings, et~al.
\newblock Advances and open problems in federated learning.
\newblock {\em arXiv preprint arXiv:1912.04977}, 2019.

\bibitem{KingBa15}
D.~Kingma and J.~Ba.
\newblock Adam: A method for stochastic optimization.
\newblock In {\em ICLR}, 2015.

\bibitem{FederatedScope-LLM}
W.~Kuang, B.~Qian, Z.~Li, D.~Chen, D.~Gao, X.~Pan, Y.~Xie, Y.~Li, B.~Ding, and J.~Zhou.
\newblock Federatedscope-llm: A comprehensive package for fine-tuning large language models in federated learning.
\newblock In {\em Proceedings of the 30th ACM SIGKDD Conference on Knowledge Discovery and Data Mining}, KDD '24, page 5260–5271, New York, NY, USA, 2024. Association for Computing Machinery.

\bibitem{PytorchDistributed}
S.~Li, Y.~Zhao, R.~Varma, O.~Salpekar, P.~Noordhuis, T.~Li, A.~Paszke, J.~Smith, B.~Vaughan, P.~Damania, and S.~Chintala.
\newblock Pytorch distributed: experiences on accelerating data parallel training.
\newblock {\em Proc. VLDB Endow.}, 13(12):3005–3018, aug 2020.

\bibitem{sgdm}
Y.~Liu, Y.~Gao, and W.~Yin.
\newblock An improved analysis of stochastic gradient descent with momentum.
\newblock {\em NeurIPS}, 33:18261--18271, 2020.

\bibitem{mcmahan2017communication}
B.~McMahan, E.~Moore, D.~Ramage, S.~Hampson, and B.~A. y~Arcas.
\newblock Communication-efficient learning of deep networks from decentralized data.
\newblock In {\em Artificial Intelligence and Statistics}, pages 1273--1282, 2017.

\bibitem{wikitext}
S.~Merity, C.~Xiong, J.~Bradbury, and R.~Socher.
\newblock Pointer sentinel mixture models.
\newblock In {\em ICLR}, 2017.

\bibitem{mishchenko2022proxskip}
K.~Mishchenko, G.~Malinovsky, S.~Stich, and P.~Richt{\'a}rik.
\newblock Proxskip: Yes! local gradient steps provably lead to communication acceleration! finally!
\newblock In {\em ICML}, pages 15750--15769, 2022.

\bibitem{narayanan2021efficient}
D.~Narayanan, M.~Shoeybi, J.~Casper, P.~LeGresley, M.~Patwary, V.~Korthikanti, D.~Vainbrand, P.~Kashinkunti, J.~Bernauer, B.~Catanzaro, et~al.
\newblock Efficient large-scale language model training on gpu clusters using megatron-lm.
\newblock In {\em SC}, 2021.

\bibitem{chatgpt}
OpenAI.
\newblock Introducing chatgpt.
\newblock \url{https://openai.com/blog/chatgpt}, 2022.

\bibitem{dataexhaust}
M.~Parekh.
\newblock Ai: Our data exhaust.
\newblock \url{https://medium.com/@mparekh/ai-our-data-exhaust-e08ba5885af8}.

\bibitem{coltLocalSGD}
K.~K. Patel, M.~Glasgow, A.~Zindari, L.~Wang, S.~U. Stich, Z.~Cheng, N.~Joshi, and N.~Srebro.
\newblock The limits and potentials of local sgd for distributed heterogeneous learning with intermittent communication.
\newblock In {\em COLT}, volume 247, pages 4115--4157, 2024.

\bibitem{10.1145/3341301.3359642}
Y.~Peng, Y.~Zhu, Y.~Chen, Y.~Bao, B.~Yi, C.~Lan, C.~Wu, and C.~Guo.
\newblock A generic communication scheduler for distributed dnn training acceleration.
\newblock In {\em SOSP}, page 16–29, 2019.

\bibitem{qin2024federated}
Z.~Qin, D.~Chen, B.~Qian, B.~Ding, Y.~Li, and S.~Deng.
\newblock Federated full-parameter tuning of billion-sized language models with communication cost under 18 kilobytes.
\newblock In {\em Forty-first International Conference on Machine Learning}, 2024.

\bibitem{gpt}
A.~Radford, K.~Narasimhan, T.~Salimans, and I.~Sutskever.
\newblock Improving language understanding by generative pre-training.
\newblock {\em URL https://s3-us-west-2. amazonaws. com/openai-assets/research-covers/language-unsupervised/language\_understanding\_paper. pdf}, 2018.

\bibitem{gpt2}
A.~Radford, J.~Wu, R.~Child, D.~Luan, D.~Amodei, I.~Sutskever, et~al.
\newblock Language models are unsupervised multitask learners.
\newblock {\em OpenAI blog}, 1(8):9, 2019.

\bibitem{SVCCA}
M.~Raghu, J.~Gilmer, J.~Yosinski, and J.~Sohl-Dickstein.
\newblock Svcca: singular vector canonical correlation analysis for deep learning dynamics and interpretability.
\newblock In {\em NeurIPS}, 2017.

\bibitem{Rajbhandari2019ZeROMO}
S.~Rajbhandari, J.~Rasley, O.~Ruwase, and Y.~He.
\newblock Zero: Memory optimizations toward training trillion parameter models.
\newblock {\em SC20: International Conference for High Performance Computing, Networking, Storage and Analysis}, pages 1--16, 2019.

\bibitem{reid2024gemini}
M.~e.~a. Reid.
\newblock Gemini 1.5: Unlocking multimodal understanding across millions of tokens of context.
\newblock {\em arXiv preprint arXiv:2403.05530}, 2024.

\bibitem{ren2021zero}
J.~Ren, S.~Rajbhandari, R.~Y. Aminabadi, O.~Ruwase, S.~Yang, M.~Zhang, D.~Li, and Y.~He.
\newblock $\{$Zero-offload$\}$: Democratizing $\{$billion-scale$\}$ model training.
\newblock In {\em 2021 USENIX Annual Technical Conference (USENIX ATC 21)}, pages 551--564, 2021.

\bibitem{ryabinin2023swarm}
M.~Ryabinin, T.~Dettmers, M.~Diskin, and A.~Borzunov.
\newblock Swarm parallelism: training large models can be surprisingly communication-efficient.
\newblock In {\em ICML}, 2023.

\bibitem{DecentMOE}
M.~Ryabinin and A.~Gusev.
\newblock Towards crowdsourced training of large neural networks using decentralized mixture-of-experts.
\newblock In H.~Larochelle, M.~Ranzato, R.~Hadsell, M.~Balcan, and H.~Lin, editors, {\em NeurIPS}, volume~33, pages 3659--3672. Curran Associates, Inc., 2020.

\bibitem{sadiev2022communication}
A.~Sadiev, D.~Kovalev, and P.~Richt{\'a}rik.
\newblock Communication acceleration of local gradient methods via an accelerated primal-dual algorithm with an inexact prox.
\newblock {\em NeurIPS}, 35:21777--21791, 2022.

\bibitem{sani2024photon}
L.~Sani, A.~Iacob, Z.~Cao, R.~Lee, B.~Marino, Y.~Gao, D.~Cai, Z.~Li, W.~Zhao, X.~Qiu, et~al.
\newblock Photon: Federated llm pre-training.
\newblock {\em arXiv preprint arXiv:2411.02908}, 2024.

\bibitem{9488803}
S.~Shi, X.~Chu, and B.~Li.
\newblock Exploiting simultaneous communications to accelerate data parallel distributed deep learning.
\newblock In {\em IEEE INFOCOM}, pages 1--10, 2021.

\bibitem{9155269}
S.~Shi, Q.~Wang, X.~Chu, B.~Li, Y.~Qin, R.~Liu, and X.~Zhao.
\newblock Communication-efficient distributed deep learning with merged gradient sparsification on gpus.
\newblock In {\em IEEE INFOCOM}, 2020.

\bibitem{shi2019distributed}
S.~Shi, Q.~Wang, K.~Zhao, Z.~Tang, Y.~Wang, X.~Huang, and X.~Chu.
\newblock A distributed synchronous sgd algorithm with global top-k sparsification for low bandwidth networks.
\newblock In {\em 2019 IEEE 39th International Conference on Distributed Computing Systems (ICDCS)}, pages 2238--2247. IEEE, 2019.

\bibitem{stich2018localsgd}
S.~U. Stich.
\newblock Local {SGD} converges fast and communicates little.
\newblock In {\em ICLR}, 2019.

\bibitem{10621164}
N.~Su, C.~Hu, B.~Li, and B.~Li.
\newblock Titanic: Towards production federated learning with large language models.
\newblock In {\em IEEE INFOCOM}, 2024.

\bibitem{sunco2}
W.~Sun, Z.~Qin, W.~Sun, S.~Li, D.~Li, X.~Shen, Y.~Qiao, and Y.~Zhong.
\newblock Co2: Efficient distributed training with full communication-computation overlap, 2024.

\bibitem{tang2024fusionllmdecentralizedllmtraining}
Z.~Tang, X.~Kang, Y.~Yin, X.~Pan, Y.~Wang, X.~He, Q.~Wang, R.~Zeng, K.~Zhao, S.~Shi, A.~C. Zhou, B.~Li, B.~He, and X.~Chu.
\newblock Fusionllm: A decentralized llm training system on geo-distributed gpus with adaptive compression, 2024.

\bibitem{tang2020survey}
Z.~Tang, S.~Shi, X.~Chu, W.~Wang, and B.~Li.
\newblock Communication-efficient distributed deep learning: A comprehensive survey.
\newblock {\em arXiv preprint arXiv:2003.06307}, 2020.

\bibitem{GossipFL}
Z.~Tang, S.~Shi, B.~Li, and X.~Chu.
\newblock Gossipfl: A decentralized federated learning framework with sparsified and adaptive communication.
\newblock {\em IEEE Transactions on Parallel and Distributed Systems}, pages 1--13, 2022.

\bibitem{tang2023fusionai}
Z.~Tang, Y.~Wang, X.~He, L.~Zhang, X.~Pan, Q.~Wang, R.~Zeng, K.~Zhao, S.~Shi, B.~He, et~al.
\newblock Fusionai: Decentralized training and deploying llms with massive consumer-level gpus.
\newblock In {\em IJCAI Symposium on LLMs}, 2023.

\bibitem{tang2024fusefl}
Z.~Tang, Y.~Zhang, P.~Dong, Y.~ming Cheung, A.~C. Zhou, B.~Han, and X.~Chu.
\newblock Fusefl: One-shot federated learning through the lens of causality with progressive model fusion.
\newblock In {\em The Thirty-eighth Annual Conference on Neural Information Processing Systems}, 2024.

\bibitem{VHL}
Z.~Tang, Y.~Zhang, S.~Shi, X.~He, B.~Han, and X.~Chu.
\newblock Virtual homogeneity learning: Defending against data heterogeneity in federated learning.
\newblock In {\em ICML}, volume 162, pages 21111--21132, 2022.

\bibitem{llm_medicine}
A.~J. Thirunavukarasu, D.~S.~J. Ting, K.~Elangovan, L.~Gutierrez, T.~F. Tan, and D.~S.~W. Ting.
\newblock Large language models in medicine.
\newblock {\em Nature medicine}, 29(8):1930--1940, 2023.

\bibitem{touvron2023llama}
H.~Touvron et~al.
\newblock Llama 2: Open foundation and fine-tuned chat models.
\newblock {\em arXiv preprint arXiv:2307.09288}, 2023.

\bibitem{Touvron2023LLaMAOA}
H.~Touvron, T.~Lavril, G.~Izacard, X.~Martinet, M.-A. Lachaux, T.~Lacroix, B.~Rozi{\`e}re, N.~Goyal, E.~Hambro, F.~Azhar, A.~Rodriguez, A.~Joulin, E.~Grave, and G.~Lample.
\newblock Llama: Open and efficient foundation language models.
\newblock {\em ArXiv}, 2023.

\bibitem{villalobos2022will}
P.~Villalobos, J.~Sevilla, L.~Heim, T.~Besiroglu, M.~Hobbhahn, and A.~Ho.
\newblock Will we run out of data? an analysis of the limits of scaling datasets in machine learning.
\newblock {\em arXiv e-prints}, pages arXiv--2211, 2022.

\bibitem{wang2019adaptive}
J.~Wang and G.~Joshi.
\newblock Adaptive communication strategies to achieve the best error-runtime trade-off in local-update sgd.
\newblock In {\em MLsys}, 2019.

\bibitem{10229086}
J.~Wang, B.~Liang, M.~Dong, G.~Boudreau, and A.~Afana.
\newblock Online distributed optimization with efficient communication via temporal similarity.
\newblock In {\em INFOCOM}, pages 1--10, 2023.

\bibitem{cocktailSGD}
J.~Wang, Y.~Lu, B.~Yuan, B.~Chen, P.~Liang, C.~De~Sa, C.~Re, and C.~Zhang.
\newblock {C}ocktail{SGD}: Fine-tuning foundation models over 500{M}bps networks.
\newblock In {\em ICML}. PMLR, 2023.

\bibitem{NEURIPS2022_7a43b8eb}
J.~Wang, B.~Yuan, L.~Rimanic, Y.~He, T.~Dao, B.~Chen, C.~R\'{e}, and C.~Zhang.
\newblock Fine-tuning language models over slow networks using activation quantization with guarantees.
\newblock In {\em NeurIPS}, 2022.

\bibitem{woodworth2020minibatch}
B.~E. Woodworth, K.~K. Patel, and N.~Srebro.
\newblock Minibatch vs local sgd for heterogeneous distributed learning.
\newblock {\em NeurIPS}, 33:6281--6292, 2020.

\bibitem{wu2023bloomberggpt}
S.~Wu, O.~Irsoy, S.~Lu, V.~Dabravolski, M.~Dredze, S.~Gehrmann, P.~Kambadur, D.~Rosenberg, and G.~Mann.
\newblock Bloomberggpt: A large language model for finance.
\newblock {\em arXiv preprint arXiv:2303.17564}, 2023.

\bibitem{298559}
M.~Xu, D.~Cai, Y.~Wu, X.~Li, and S.~Wang.
\newblock {FwdLLM}: Efficient federated finetuning of large language models with perturbed inferences.
\newblock In {\em 2024 USENIX Annual Technical Conference (USENIX ATC 24)}, pages 579--596, 2024.

\bibitem{openfedllm}
R.~Ye, W.~Wang, J.~Chai, D.~Li, Z.~Li, Y.~Xu, Y.~Du, Y.~Wang, and S.~Chen.
\newblock Openfedllm: Training large language models on decentralized private data via federated learning.
\newblock In {\em Proceedings of the 30th ACM SIGKDD Conference on Knowledge Discovery and Data Mining}, KDD '24, page 6137–6147, 2024.

\bibitem{ye2024galaxy}
S.~Ye, J.~Du, L.~Zeng, W.~Ou, X.~Chu, Y.~Lu, and X.~Chen.
\newblock Galaxy: A resource-efficient collaborative edge ai system for in-situ transformer inference.
\newblock In {\em INFOCOM}, 2024.

\bibitem{Asteroid}
S.~Ye, L.~Zeng, X.~Chu, G.~Xing, and X.~Chen.
\newblock Asteroid: Resource-efficient hybrid pipeline parallelism for collaborative dnn training on heterogeneous edge devices.
\newblock In {\em ACM MobiCom}, page 312–326, 2024.

\bibitem{yuandecentralized}
B.~Yuan, Y.~He, J.~Davis, T.~Zhang, T.~Dao, B.~Chen, P.~S. Liang, C.~Re, and C.~Zhang.
\newblock Decentralized training of foundation models in heterogeneous environments.
\newblock {\em Advances in Neural Information Processing Systems}, 35:25464--25477, 2022.

\bibitem{AperiodcLocalSGD}
H.~Zhang, T.~Wu, S.~Cheng, and J.~Liu.
\newblock Aperiodic local sgd: Beyond local sgd.
\newblock In {\em ICPP}, 2023.

\bibitem{203269}
H.~Zhang, Z.~Zheng, S.~Xu, W.~Dai, Q.~Ho, X.~Liang, Z.~Hu, J.~Wei, P.~Xie, and E.~P. Xing.
\newblock Poseidon: An efficient communication architecture for distributed deep learning on {GPU} clusters.
\newblock In {\em USENIX ATC}, pages 181--193, Santa Clara, CA, July 2017.

\bibitem{zhuang2023foundation}
W.~Zhuang, C.~Chen, and L.~Lyu.
\newblock When foundation model meets federated learning: Motivations, challenges, and future directions.
\newblock {\em arXiv preprint arXiv:2306.15546}, 2023.

\end{thebibliography}

\clearpage
\appendix 
\section*{Appendix}\label{apdx:proof}

In this appendix, we present the detailed proof of Theorem~\ref{theo:converge-bound}.

\begin{lemma}
\label{lemma:b1}
   \textbf{Convergence of  averaged weights. } Let $\{w_r\}_{r\geq0}$ and $\{\Bar{w}_r\}_{r\geq0}$ for $k\in[K]$ be defined as in Equation (\ref{eq:partial-LSGD}),  and let f be $\beta$-smooth and $\mu$-strongly convex and $\eta_r\leq\frac{1}{4\beta}$. Then
    \begin{equation*}
    \small
    \begin{split}
        \mathbb{E}||\Bar{w}_{r+1}-w^*||^2
        &\leq(1-\mu\eta_r)\mathbb{E}||\Bar{w}_r-w^*||^2+\eta_r^2\mathbb{E}||g_r-\Bar{g}_r||^2\\
        &-\frac{1}{2}\eta_r\mathbb{E}(f(\Bar{w}_r)-f^*)+2\eta_r\frac{\beta}{K}\sum_{k=1}^K\mathbb{E}||\Bar{w}_r-w_r^k||^2
    \end{split}
    \end{equation*}
\end{lemma}
\begin{proof}[Proof of Lemma~\ref{lemma:b1}]
Using the update Eq.~\ref{eq:partial-LSGD}, we have 
\begin{equation}
\small
\label{eq:b0}
\begin{split}
||\Bar{w}_{r+1}-w^*||^2=&||\Bar{w}_r-\eta_rg_r-w^*-\eta_r\Bar{g}_r+\eta_r\Bar{g}_r||^2\\
    =&||\Bar{w}_r-\eta_r\Bar{g}_r-w^*||^2+\eta_r^2||g_r-\Bar{g}_r||^2\\
    &+2\eta_r\langle\Bar{w}_r-w^*-\eta_r\Bar{g}_r,g_r-\Bar{g}_r\rangle.
\end{split}
\end{equation}
By $L$-smoothness and $\mu$-strong convexity of $f$, we can derive the upper bound of the first term as 
\begin{equation}
\small
\label{eq:b4}
    \begin{split}
        &||\Bar{w}_r-w^*-\eta_r\Bar{g}_r||^2\leq(1-\mu\eta_r)||\Bar{w}_r-w^*||^2\\
        &-\frac{1}{2}\eta_r(f(\Bar{w}_r)-f^*)+\frac{2\eta_rL}{M}\sum_{k=1}^K||\Bar{w}_r-w_r^k||^2.
    \end{split}
\end{equation}

Taking (\ref{eq:b4}) back into (\ref{eq:b0}), by taking expectation we have 
\begin{equation}
\label{eq:b5}
\small
\begin{split}
\mathbb{E}||\Bar{w}_{r+1}-w^*||^2 \leq&(1-\mu\eta_r)\mathbb{E}||\Bar{w}_r-w^*||^2+\eta_r^2\mathbb{E}||g_r-\Bar{g}_r||^2\\
    &-\frac{1}{2}\eta_r\mathbb{E}(f(\Bar{w}_r)-f^*)+\frac{2L\eta_r }{M}\sum_{k=1}^K\mathbb{E}||\Bar{w}_r-w_{r}^{k}||^2
\end{split}
\end{equation}
\end{proof}

\begin{lemma}
\label{lemma:gap}
\textbf{Bounded model divergence.} Given the bounded synchronization frequency $H_l\leq H$ for $l\in[L]$ and sequence of decreasing positive learning rates $\{\eta_r\}_{r\geq0}$ satisfying $\eta_r\leq2\eta_{r+H_l}$ for all $r\geq0$, then we have
    \begin{equation}
        \frac{1}{K}\sum_{k=1}^K\mathbb{E}||\Bar{w}_r-w_r^k||^2\leq4H^2\eta_r^2G^2.
    \end{equation}
\end{lemma}

\begin{proof}[Proof of Lemma~\ref{lemma:gap}]
As $w_r^k$ is the concatenation of $w_r^{k,l}$ with respect to $l$, we have $||\Bar{w}_r-w_r^k||^2=\sum_{l=1}^L||\Bar{w}_r^l-w_r^{k,l}||^2$, $\sum_{l=1}^L||\nabla^l f(w_h^k;\xi_h^k)||^2=||\nabla f(w_h^k;\xi_h^k)||^2$. Using $\mathbb{E}||X-\mathbb{E}X||^2=\mathbb{E}||X||^2-||\mathbb{E}X||^2$ and $||\sum_{i=1}^{H_l}a_i||^2\leq H_l\sum_{i=1}^{H_l}||a_i||^2$, we have
\begin{equation*}
\small
    \begin{split}
    \frac{1}{K}\sum_{k=1}^K\mathbb{E}||\Bar{w}_r-w_r^k||^2
        &=\frac{1}{K}\sum_{k=1}^K\sum_{l=1}^L\mathbb{E}||w_{r}^{k,l}-w_{r_0}^{k,l}-(\Bar{w}_{r}^l-w_{r_0}^{k,l})||^2\\
        &\leq\frac{1}{K}\sum_{k=1}^K\sum_{l=1}^L H_l \eta_{r_0}^2\sum_{h=r_0}^{r-1}\mathbb{E}||\nabla^l f(w_h^k;\xi_h^k)||^2\\
        &\leq\frac{1}{K}\sum_{k=1}^K H \eta_{r_0}^2\sum_{h=r_0}^{r-1}\mathbb{E}||\nabla f(w_h^k;\xi_h^k)||^2\\
        &\leq4H^2\eta_r^2G^2.
    \end{split}
\end{equation*}
\end{proof}

\begin{lemma}
\label{lemma:supermartingale}
\textbf{Martingale convergence.}    Let $\{a_r\}_{r\geq0}, a_r\geq0,\{e_r\}_{r\geq0},e_r\geq0$ be sequences satisfying
\begin{equation}
\small
    a_{r+1}\leq(1-\mu\eta_r)a_r-\eta_re_rA+\eta_r^2B+\eta_r^3C,
\end{equation}
for $\eta_r=\frac{4}{\mu(a+r)}$ and constants $A>0,B,C\geq0,\mu>0,a>1$. Then
\begin{equation}
\small
    \frac{A}{S_R}\sum_{r=1}^{R-1}p_re_r\leq\frac{\mu a^3}{4S_R}a_0+\frac{2R(R+2a)}{\mu S_R}B+\frac{16R}{\mu^2S_R}C,
\end{equation}
for $p_r=(a+r)^2$ and $S_R\triangleq\sum_{r=0}^{R-1}p_r=\frac{R}{6}(2R^2+6aR-3R+6a^2-6a+1)\geq\frac{1}{3}R^3$.
\end{lemma}

\begin{proof}[Proof of Theorem~\ref{theo:converge-bound}]
Based on Lemma~\ref{lemma:b1} and Lemma~\ref{lemma:gap} we get
\begin{equation}
\small
\begin{split}
    \mathbb{E}||\Bar{w}_{r+1}-w^*||^2
    \leq&(1-\mu\eta_r)\mathbb{E}||\Bar{w}_r-w^*||^2+\frac{\sigma^2}{K}\eta_r^2\\
    &-\frac{1}{2}\eta_r\mathbb{E}(f(\Bar{w}_r)-f^*)+8G^2H^2\beta\eta_r^3.
\end{split}
\label{eq:converge-seq}
\end{equation}
By Lemma~\ref{lemma:supermartingale}, taking $a_r=\mathbb{E}||\Bar{w}_r-w^*||^2, e_r=\Ebb(f(\Bar{w}_r)-f^*), A=\frac{1}{2}, B=\frac{\sigma^2}{K} ,C=8G^2H^2\beta$, we have
\begin{equation}
    \small
    \label{eq:proof-marting}
    \begin{split}
          \frac{1}{2S_R}\sum_{r=1}^R p_r\Ebb(f(\Bar{w}_r)-f^*)&\leq \frac{\mu a^3}{4S_R}\mathbb{E}||\Bar{w}_0-w^*||^2\\
          &+\frac{2R\sigma^2(R+2a)}{\mu KS_R}+\frac{128G^2H^2R}{\mu^2S_R}.
    \end{split}
\end{equation}
By the convexity of $f$, we have 
\begin{equation}
\label{eq:proof-convexity}
    \mathbb{E}f(\hat{w}_R)-f^*\leq\frac{1}{S_R}\sum_{r=1}^R\mathbb{E}(f(\Bar{w}_r)-f^*)
\end{equation}
By Eq.~\ref{eq:proof-marting},~\ref{eq:proof-convexity} we get the result
\begin{equation}
    \small
    \begin{split}
        \mathbb{E}f(\hat{w}_R)-f^*&\leq \frac{\mu a^3}{2S_R}||w_0-w^*||^2\\
        &+\frac{4R(R+2a)}{\mu K S_R}\sigma^2+\frac{256R}{\mu^2S_R}G^2H^2\beta,
    \end{split}
\end{equation}
which completes the proof of Theorem~\ref{theo:converge-bound}.
\end{proof}




\end{document}